\documentclass[journal]{IEEEtran}
\IEEEoverridecommandlockouts
 \usepackage[english]{babel}

\newcommand{\avec}{{\bf{a}}}

\newcommand{\bvec}{{\bf{b}}}
\newcommand{\cvec}{{\bf{c}}}

\newcommand{\evec}{{\bf{e}}}
\newcommand{\fvec}{{\bf{f}}}

\newcommand{\qvec}{{\bf{q}}}

\newcommand{\yvec}{{\bf{y}}}
\newcommand{\uvec}{{\bf{u}}}

\newcommand{\xvec}{{\bf{x}}}

\newcommand{\vvec}{{\bf{v}}}

\newcommand{\hvec}{{\bf{h}}}

\newcommand{\onevec}{{\bf{1}}}
\newcommand{\zerovec}{{\bf{0}}}

\newcommand{\Lambdamat}{{\bf{\Lambda}}}

\newcommand{\Amat}{{\bf{A}}}
\newcommand{\Bmat}{{\bf{B}}}

\newcommand{\Dmat}{{\bf{D}}}

\newcommand{\Fmat}{{\bf{F}}}

\newcommand{\Hmat}{{\bf{H}}}

\newcommand{\Imat}{{\bf{I}}}
\newcommand{\Lmat}{{\bf{L}}}

\newcommand{\Pmat}{{\bf{P}}}
\newcommand{\Mmat}{{\bf{M}}}

\newcommand{\Qmat}{{\bf{Q}}}

\newcommand{\Smat}{{\bf{S}}}

\newcommand{\Rmat}{{\bf{R}}}

\newcommand{\Vmat}{{\bf{V}}}

\newcommand{\Sigmamat}{{\bf{\Sigma}}}

\newcommand{\define}{\stackrel{\triangle}{=}}

\newcommand{\be}{\begin{equation}}
\newcommand{\ee}{\end{equation}}
\newcommand{\beqna}{\begin{eqnarray}}
\newcommand{\eeqna}{\end{eqnarray}}

\usepackage{cite}
\usepackage{amsmath,amssymb,amsfonts}
\usepackage{graphicx}
\usepackage{xcolor}
\usepackage{subcaption}
\usepackage{acronym}
\let\labelindent\relax 
\usepackage{enumitem}
\usepackage[bookmarks,colorlinks]{hyperref}
\usepackage{ntheorem}
\newtheorem{theorem}{Theorem}

 \newtheorem{Claim}{Claim}
  \newtheorem{Lemma}{Lemma}
  \newtheorem{corollary}{Corollary}
\usepackage{subcaption}
\DeclareMathAlphabet{\mathbfcal}{OMS}{cmsy}{b}{n}
\DeclareMathAlphabet{\pazocal}{OMS}{zplm}{m}{n}

\usepackage[ruled,vlined]{algorithm2e}
\SetKwInput{KwInput}{Input}                
\SetKwInput{KwOutput}{Output}              
\newcommand{\diag}{\mathop{\mathrm{diag}}}

\newcommand\maxEdges{\frac{N(N-1)}{2}}
\newcommand\timeIndx{l}

\acrodef{admm}[ADMM]{alternating
direction method of multiplier}
\acrodef{ac}[AC]{alternating current}
\acrodef{arma}[ARMA]{auto-regressive moving average}
\acrodef{cfar}[CFAR]{constant false alarm rate}
\acrodef{crb}[CRB]{ Cram$\acute{\text{e}}$r-Rao  bound}
\acrodef{cs}[CS]{compressed sensing}
\acrodef{dau}[DAU]{deep algorithm unrolling}
\acrodef{dc}[DC]{direct current}
\acrodef{dnn}[DNN]{deep neural network}
\acrodef{dsp}[DSP]{digital signal processing}
\acrodef{dl}[DL]{Deep Learning}
\acrodef{ems}[EMS]{energy management system}
\acrodef{ekf}[EKF]{extended \ac{kf}}
\acrodef{evd}[EVD]{eigenvalue decomposition}
\acrodef{eier}[EIER]{edge
identification error rate}
\acrodef{gans}[GANs]{generative
adversarial networks}
\acrodef{gso}[GSO]{Graph shift operator}
\acrodef{gsp}[GSP]{graph signal processing}
\acrodef{gnn}[GNNs]{graph neural networks}
\acrodef{glrt}[GLRT]{generalized likelihood ratio test}
\acrodef{gso}[GSO]{graph shift operator}
\acrodef{gft}[GFT]{graph Fourier transform}
\acrodef{gmrf}[GMRF]{Gaussian Markov random field}
\acrodef{gcs}[GCS]{graph convolutional skip}
\acrodef{gcn}[GCNs]{graph convolutional networks}
\acrodef{gwkf}[GW-KF]{graph weights Kalman filter}

\acrodef{gtkf}[GT-KF]{graph topology Kalman filter}
\acrodefplural{gtkf}[GT-KFs]{graph topology Kalman filters}
\acrodef{iid}[i.i.d.]{independent and identically distributed}
\acrodef{ista}[ISTA]{iterative soft thresholding algorithm}
\acrodef{kg}[KG]{Kalman gain}
\acrodef{kkt}[KKT]{Karush-Kuhn-Tucker}
\acrodef{kf}[KF]{Kalman filter}
\acrodef{lrt}[LRT]{likelihood ratio test}
\acrodef{lista}[LISTA]{learned \ac{ista}}
\acrodef{lmmse}[LMMSE]{linear minimum mean-squared error}
\acrodef{map}[MAP]{{\em{maximum a-posteriori}} probability}
\acrodef{ml}[ML]{maximum likelihood}
\acrodef{mse}[MSE]{mean-squared-error}
\acrodef{mmse}[MMSE]{minimum MSE}
\acrodef{nns}[NNs]{neural networks}
\acrodef{pdf}[PDF]{probability density function}
\acrodef{pf}[PF]{power flow}
\acrodef{pgd}[PGD]{proximal gradient descent}
\acrodef{pfa}[PFA]{probability of false alarm}
\acrodef{psd}[PSD]{positive semi-definite}
\acrodef{psse}[PSSE]{power systems state estimation}
\acrodef{pmu}[PMUs]{phasor measurement units}
\acrodef{roc}[ROC]{Receiver Operating Characteristic}
\acrodef{snr}[SNR]{signal-to-noise ratio}
\acrodef{svd}[SVD]{singular value decomposition}
\acrodef{svar}[SVAR]{structural vector autoregressive
model}
\acrodef{ss}[SS]{state space}
\acrodef{ssm}[SSM]{state-space model}
\acrodef{tv}[TV]{total variation}
\acrodef{wls}[WLS]{weighted least squares}
\acrodef{wrt}[w.r.t.]{with respect to}

\newcommand{\myVec}[1]{{\boldsymbol{#1}}}

\newcommand{\Input}{\myVec{x}}

\newcommand{\Kgain}{\mathbfcal{K}}

\title{Sparsity-Aware Extended Kalman Filter for Tracking \\  Dynamic Graphs}
\author{Lital Dabush, \IEEEmembership{Student Member, IEEE}, 
Nir Shlezinger, \IEEEmembership{Senior Member, IEEE},  and
Tirza Routtenberg, \IEEEmembership{Senior Member, IEEE}
\thanks{{\footnotesize{The authors are with the School of Electrical and Computer Engineering, Ben-Gurion University of the Negev, Beer-Sheva 84105, Israel, e-mail: litaldab@post.bgu.ac.il, \{nirshl; tirzar\}@bgu.ac.il.}}}
\thanks{Preliminary parts of this work were presented at the IEEE International Conference on Acoustics, Speech, and Signal Processing (ICASSP) 2023 as the paper~\cite{dabush2024Routtenberg_Kalman}. This research was supported by the ISRAEL SCIENCE FOUNDATION (Grant No. 1148/22) and by the Israel Ministry of National Infrastructure and Energy.  L. Dabush is a fellow of the AdR Women Doctoral Program.
 }
 \vspace{-0.5cm}
}
\graphicspath{{figures/}}
\begin{document}
    \maketitle    
\begin{abstract}
A broad range of applications involve signals with irregular structures that can be represented as a graph. As the underlying structures can change over time, the tracking dynamic graph topologies from observed signals is a fundamental challenge in graph signal processing (GSP), with applications in various domains, such as power systems, the brain-machine interface, and communication systems.
In this paper, we propose a method for tracking dynamic changes in graph topologies.  Our approach builds on a representation of the dynamics as a graph-based nonlinear state-space model (SSM), where the observations are graph signals generated through graph filtering,
and the underlying evolving topology serves as the latent states. 
In our formulation, the graph Laplacian matrix is parameterized using the incidence matrix and edge weights, enabling a structured representation of the state. 
In order to track the evolving topology in the resulting SSM, we develop a sparsity-aware extended Kalman filter (EKF) that integrates $\ell_1$-regularized updates within the filtering process.  Furthermore, a dynamic programming scheme to efficiently compute the Jacobian of the graph filter is introduced.
Our numerical study demonstrates the ability of the proposed method to accurately track sparse and time-varying graphs under realistic conditions,  with highly nonlinear measurements, various noise levels, and different change rates, while maintaining low computational complexity.
\end{abstract}

\section{Introduction}

Modern complex systems in various fields, ranging from engineering and physics to biology and sociology, involve data supported on irregular structures that are naturally modeled as graph signals \cite{newman2018networks}. 
The study of the signals, indexed by the vertices of a graph, has led to the development of \ac{gsp}, which generalizes classical \ac{dsp} theory to data defined on graphs \cite{Shuman_Ortega_2013,Sandryhaila2014,8347162}. While \ac{gsp} studies usually focus on settings where the graphs are fixed, in many practical scenarios the topology of the graph can evolve over time. This feature is critical for applications such as power grid monitoring, gas and water distribution, and transportation networks,  where efficient tracking of dynamic changes is vital for ensuring effective system operation \cite{shaked2021identification,gas_dis_system}.

Temporal variations associated with graphs and graph signals motivate their representation as a time series~\cite{chouzenoux2023sparse}. A fundamental approach used for probabilistic and interpretable time-series estimation is the \ac{kf} and its extended version \cite{gannot2008kalman}. The \ac{ekf}, which is based on a statistical representation of the dynamics as a nonlinear \ac{ssm} with Gaussian noise signals~\cite{durbin2012time}, relies on local linearization of the \ac{ssm}. \ac{ssm}-based tracking methods, such as the \ac{ekf}, have demonstrated remarkably good performance and computational efficiency in tracking and control tasks, and have been adopted for tasks involving the {\em tracking of graph signals}.  For example, the \ac{kf} was used for distributed state estimation over sensor networks \cite{Ling_2009}, 
tracking network traffic \cite{Soule2005}, graph-time filters \cite{isufi20162}, tracking bandlimited graph signals \cite{Isufi2020}, and tracking random graph processes from a nonlinear dynamic using \ac{gsp} tools \cite{Sagi2023Routtenberg,buchnik2023gspkalmannet}. The works in \cite{Ling_2009,Soule2005,isufi20162,Isufi2020,Sagi2023Routtenberg,buchnik2023gspkalmannet} focus on graphs with known topologies, while \ac{ssm}-based tracking has also been proposed for settings where the topologies are unknown. 
In this context, Bayesian approaches for estimating the state transition matrix and the noise covariance have been proposed \cite{Cox2023Elvira, chouzenoux2023sparse}. An expectation-maximization algorithm for estimating both the transition matrix and the states has been proposed in \cite{elvira2022graphical}. In \cite{Cox2025Polynomial}, a particle filter that approximates the state dynamics and the observation model using polynomials of the state has been proposed. However, in all these works, the graph topology is assumed to be stationary, and tracking is performed on the {\em graph signals}. 

Despite the importance of time-varying graphs, most existing tools for learning graph topologies assume stationary signals. These include works that learn the graph Laplacian via inverse covariance estimation with sparsity regularization \cite{lake2010discovering,medvedovsky2023efficient,7979524,NEURIPS2019_90cc440b} or by imposing smoothness of graph signals \ac{wrt} the topology \cite{Dong_Vandergheynst_2016,Vassilis_2016}. However, such approaches assume a fixed, time-invariant topology. Most existing works that account for the dynamics of the graph consider linear (first-order)  topologies, including techniques based on a time-varying Graphical Lasso formulation~ \cite{hallac2017network}, and on \ac{svar} models \cite{Ioannidis2019Giannakis,Shen2018Giannakis,Money22022}. Specifically, \cite{Ioannidis2019Giannakis} jointly estimates the states and a linear topology via alternating \ac{kf} and minimization, \cite{Shen2018Giannakis}, and \cite{Money22022} extends this to nonlinear settings using kernel-based methods, where the latter also performs joint signal estimation. 
However, these methods treat all inferred connections, whether in the precision matrix or in the \ac{svar} model, as first-order neighbors, which limits their ability to accurately infer the topology in the presence of higher-order graph filtering processes.

Several works have investigated dynamic nonlinear topologies~\cite{cappelletti2024graphdictionarysignalmodelsparse, alippi2023graph,  Kalofolias2017, 10333427}, rather than the more commonly studied static topologies.
In particular,   \cite{cappelletti2024graphdictionarysignalmodelsparse} proposed to learn a dictionary of graphs that differ in edge weights, where the evolving graph topology is represented as a weighted sum of the graph atoms. 
The method in \cite{Kalofolias2017} learns independent graphs over sliding windows, incorporating temporal regularization to enforce smooth transitions in edge weights between adjacent windows. 
A method to detect topology changes between consecutive time windows has been developed in \cite{10333427}.  However, both approaches assume that changes occur slowly relative to the window length and process batches of measurements with equal weighting, limiting their ability to capture rapid dynamics or respond to recent changes. Deep-learning techniques have recently been used for tracking dynamic topologies \cite{alippi2023graph}. However, these techniques heavily depend on the availability of sufficient high-quality event labels, which are often unavailable in practice \cite{Do_Coutto_Filho2007Schilling}.

The need for reliable methods for tracking nonlinear topologies, combined with the established suitability of Kalman-type methods for tracking in nonlinear dynamics, motivates their application for dynamic graphs.
However, applying the \ac{ekf} (or other methods for nonlinear \acp{ssm} that rely on the unscented~\cite{julier2004unscented} or the cubature~\cite{arasaratnam2009cubature} transforms)  to systems with dynamic graph topologies presents significant challenges. These include the sparsity inherent in many graph structures, the computational complexity, 
and the need for a sufficient number of measurements to achieve observability. 

In this paper, we propose a graph-based \ac{ekf} for tracking dynamic graph topologies. Our algorithm is based on a new formulation of the graph transitions as a \ac{gsp} filtering model, where the measurements are graph signals and the states represent the evolving topology. 
Specifically, our approach models the evolution of graph topology through a time-varying sparse edge-weight vector, as illustrated in Fig.~\ref{fig:dynamic_graph}.
Based on this model, we introduce our adaptation of the \ac{ekf},  which incorporates sparsity-driven modeling based on the $\ell_1$ norm, along with a complexity-reduced Jacobian computation 
method. Finally, we evaluate the performance of different Kalman-filter-based methods in tracking graph topologies.

\begin{figure}
    \centering
    \includegraphics[width=0.9\columnwidth]{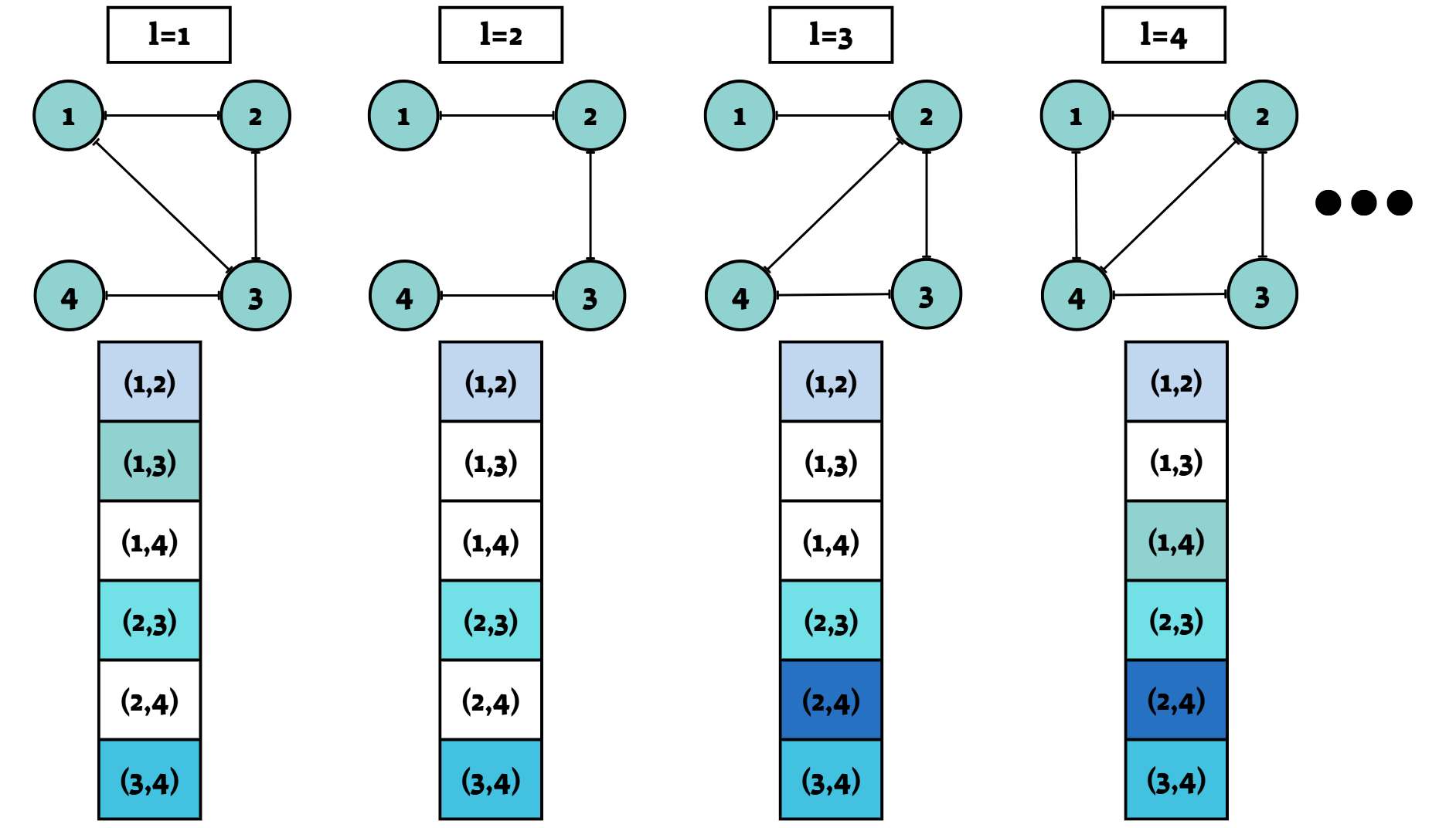}
    \vspace{-0.1cm}
    \caption{Illustration of spatiotemporal data across a simple dynamic four-node graph. 
     The top figure represents the topology as a graph; the bottom figure depicts the corresponding sparse vector representation, in which blue cells indicate edge weights, and white denotes the absence of connections (zero-weight).}\label{fig:dynamic_graph}
     \vspace{-0.25cm}
\end{figure}

Our main contributions are summarized as follows:
\begin{itemize}[leftmargin=0.25pt]
\item {\bf Dynamic Topologies as Sparse \acp{ssm}}: We present a novel sparse \ac{ssm} with an unknown support, which enables simultaneous tracking of both edge set changes and unknown edge weights. The representation accommodates a nonlinear dependence on the underlying topology. 

\item {\bf \ac{ekf}-Based Graph Tracking}: We develop a tracking algorithm for dynamic topologies based on the proposed \ac{ssm}, which combines an \ac{ekf} with sparsity regularization. To alleviate the computational burden associated with local linearization of high-dimensional nonlinear graph filters, we develop a closed-form expression for the associated Jacobian matrices of the \ac{ekf} update stage, and propose an efficient computation method based on dynamic programming.

\item {\bf Theoretical Analysis}: We present a theoretical analysis of two special cases: $(i)$ a linear measurement model, which enables an observability analysis related to the graph change rate; and $(ii)$ a known edge set (i.e., known support of the sparse vector) with unknown edge weights, which provides insights into the impact of prior structural knowledge on performance.

\item {\bf Extensive Experimentation}: We present comprehensive empirical evaluations under multiple observation models (both linear and nonlinear), analyzing the effects of graph sparsity levels, noise intensities, degrees of nonlinearity in the topology dependence, and rates of graph variation—numerically validating the insights derived from the observability analysis.
\end{itemize}

These contributions advance state-space modeling in \ac{gsp}, providing a robust framework for studying dynamic graph topologies.

{\em{Organization and Notations:}}
The remainder of this paper is structured as follows. Section~\ref{sec:System_Model_and_Preliminaries} provides the necessary preliminaries, including foundations in \ac{gsp}, \acp{ssm}, and  \ac{ekf}. In Section \ref{sec:Tracking_in_the_Graph_Frequency_Domain}, we present the proposed sparsity-aware \ac{ekf} framework for tracking dynamic graph topologies, including the problem formulation, observability analysis, algorithmic derivation, and treatment of special cases. Section \ref{ssec:simulations} introduces the numerical study setup and the simulation results. Finally, Section \ref{sec:conclusions} concludes the paper.

In the rest of this paper, vectors and matrices are denoted by boldface lowercase and uppercase letters, respectively. The notations $(\cdot)^T$ and $(\cdot)^{-1
}$ 
denote the transpose and the inverse
operators, respectively.
For a vector $\avec$, 
$a_n$ is its $n$th element,
${\operatorname{diag}}(\avec)$ is a diagonal matrix whose $(n,n)$th entry is $a_n$, the operator $\left\| \avec \right\|_{\Amat}^2 \define \avec^T \Amat \avec$ denotes the quadratic form of the matrix $\Amat$, and $\left\| \avec \right\|_p$ is the $\ell_p$ norm.   
The matrices  $\Imat$ and $\zerovec$ denote the identity and the zero matrices, respectively.

\vspace{-0.2cm}
\section{Preliminaries}
\label{sec:System_Model_and_Preliminaries}
In this section, we review some necessary preliminaries that set the ground for our derivations of the tracking algorithms in the following sections. In particular, we present the \ac{gsp} notations and graph frequency domain filtering in Subsection~\ref{ssec:Graph_Signal_Processing}. 
Then, we review the \ac{ekf}  in Subsection~\ref{ssec:EKF}.

\vspace{-0.2cm}
\subsection{Graph Signal Processing (GSP)}\label{ssec:Graph_Signal_Processing}
\subsubsection{Graph Formulation}
We denote an undirected weighted graph as ${\pazocal{G}}({\pazocal{V}},\xi)$, where ${\mathcal{V}}=\{1,\ldots,N\}$ is the  node set, and ${\xi}$ is the set of $M$ edges with positive weights. 
The set of all possible edges, corresponding to the edge set of the complete graph over ${\pazocal{V}}$, 
is denoted by $\xi_{\rm{c}}$ and contains $\frac{N(N-1)}{2}$ edges, such that $\xi \subseteq \xi_{\rm{c}}$ and $M << \frac{N(N-1)}{2}$.

To model the graph topology, we assign an arbitrary index to each possible edge  $e_m = (n, k) \in \xi_{\rm{c}}$, for any $n, k =1, \ldots, N$, with $k < n$, where $m = 1, \ldots, \maxEdges$. 
This edge connects the vertices $n$ and $k$. 
This arbitrary mapping of node pairs to edge indices is denoted by $\psi:\{\mathcal{V}\}^2\rightarrow\xi_{\rm{c}}$, where   $\psi(n,k)=m$ gives the index of edge $(n,k)$.
Correspondingly, given the nonnegative vector of weights of all edges, $\xvec\in\mathbb{R}^\maxEdges$, its $m$th entry,  $[x]_m$, is the weight of the edge $e_m$,  where  $[x]_m=0$ if $e_m\notin \xi$.

The incidence matrix of the complete graph is $\Bmat \in \mathbb{R}^{N \times \maxEdges}$, where the $(n, m)$th element of $\Bmat$ is given by \cite{newman2018networks}
\beqna
\label{Bmat_def}
B_{n,m} = \begin{cases}
1 & \text{if } e_m = (n, k) \in\xi_{\rm{c}},~~ n \text{ is the source} \\
-1 & \text{if } e_m = (k, n) \in\xi_{\rm{c}},~~ k \text{ is the source} \\
0 & \text{otherwise}
\end{cases},\hspace{-0.1cm}
\eeqna
$\forall n = 1, \ldots, N$ and $m = 1, \ldots, \maxEdges$. 
Then, we can write the graph Laplacian matrix of any graph ${\pazocal{G}}({\pazocal{V}},\xi)$ as
\beqna\label{L2B}
\Lmat \triangleq \Bmat\diag(\xvec) \Bmat^T\in \mathbb{R}^{N \times N},
\eeqna
 where the elements of $\Bmat$ are given in \eqref{Bmat_def}, and $\xvec$  contains information about both the connectivity and the weights of the edges. Correspondingly, $\xvec$ is a sparse vector, with its sparsity level determined by the percentage of existing connections relative to the total number of possible connections. 
The matrix $\Lmat$ is a real positive semidefinite matrix 
with non-positive off-diagonal elements.

\subsubsection{Graph Signals}
A graph signal is a function that 
assigns a scalar value to each node, and thus is 
an $N$-dimensional vector. A graph filter is a function $h(\cdot)$ applied to the Laplacian that allows the following \ac{evd}:
 \begin{equation} \label{laplacian_graph_filter}
  	h(\Lmat)= \Vmat h(\Lambdamat)\Vmat^T,
  \end{equation}
where the columns of $\Vmat$ are the eigenvectors of $\Lmat$, and $h(\Lambdamat)$ is a diagonal matrix. 
It should be noted that  all filters in the form of \eqref{laplacian_graph_filter} that satisfy $h(\lambda_n)=h(\lambda_m)$ for $\lambda_n=\lambda_m$, $n\neq m$ can be represented as a finite polynomial of $\Lmat$ \cite[Remark 3.5]{Ortega_2022}:
\beqna
h(\Lmat)=\sum\nolimits_{p=0}^Pa_p\Lmat^p.
\eeqna
A graph filter is applied on an input graph signal, $\avec_{\rm in}$, such that the output graph signal is \cite{Shuman_Ortega_2013}
\beqna \label{a_in_a_out}
\avec_{\rm out} = h(\Lmat)\avec_{\rm in}.
\eeqna
Graph filters are widely used to represent signal processes over networks in various applications 
\cite{Rama2020Anna,Shuman_Ortega_2013,Sandryhaila2014,8347162,dabush2023state}. 

\vspace{-0.2cm}
\subsection{Extended Kalman Filter (EKF)}
\label{ssec:EKF}
\subsubsection{State-Space Models}
 \acp{ssm} provide a conventional framework for modeling tracking tasks in regular domains (see, e.g., \cite[Ch. 10]{durbin2012time}). This framework  represents dynamic systems with a state vector $\xvec_{\timeIndx} \in \mathbb{R}^{N_x}$ and observations $\yvec_{\timeIndx} \in \mathbb{R}^{N_y}$ by characterizing their relationship as a continuous \ac{ssm} in discrete-time:
\begin{subequations}
\label{eqn:ssmodel}
\beqna\label{transition_equation}
&&\xvec_{\timeIndx} =\fvec_{\timeIndx}(\xvec_{\timeIndx-1}) + \evec_{\timeIndx},
\\\label{observation_model}
&&\yvec_{\timeIndx} = \hvec_{\timeIndx}(\xvec_{\timeIndx}) + \vvec_{\timeIndx}.
\eeqna
\end{subequations}
In \eqref{eqn:ssmodel}, $\evec_{\timeIndx}\in\mathbb{R}^{N_x}$  and $\vvec_{\timeIndx}\in\mathbb{R}^{N_y}$ are temporally \ac{iid}, zero-mean, and mutually independent state-evolution and measurement noises, respectively, at time $\timeIndx$. 
In addition, $\Qmat\in\mathbb{R}^{N_x\times N_x}$ and $\Rmat\in\mathbb{R}^{N_y\times N_y}$ are the covariance matrices of ${\evec}_{\timeIndx}$ and $\vvec_{\timeIndx}$, respectively.
The functions
$\fvec_{\timeIndx}:\mathbb{R}^{N_x}  \mapsto {\mathbb{R}}^{N_x}$  and   $\hvec_{\timeIndx}:\mathbb{R}^{N_x} \mapsto {\mathbb{R}}^{N_y}$ are (possibly nonlinear) state evolution and measurement functions, respectively.

\subsubsection{EKF}
One of the most common algorithms for state estimation in nonlinear \acp{ssm} with Gaussian noise signals is the \ac{ekf} \cite{gruber1967approach}. 
The \ac{ekf} follows the operation of the \ac{kf}, designed for linear \acp{ssm}, combining prediction based on the previous estimate with an update based on the current observation. 

The \ac{ekf} operates in two steps. The first step predicts the first- and second-order moments of the state and the observation based on the previous estimate, where the initial estimate is given by 
\begin{equation}
\label{eqn:Pred}
    \hat{\xvec}_{\timeIndx|\timeIndx-1} = \fvec_{\timeIndx}(\hat{\xvec}_{\timeIndx-1}); \quad \hat{\yvec}_{\timeIndx|\timeIndx-1} = \hvec_{\timeIndx}(\hat{\xvec}_{\timeIndx|\timeIndx-1}),
\end{equation}
and the second-order moments are propagated by 
\begin{subequations}
    \label{eqn:covariance_computation}
\begin{equation}
\label{eqn:state_covariance_computaion}
{\hat{\bf{\Sigma}}}_{\timeIndx|\timeIndx-1}=\hat\Fmat_{\timeIndx}\cdot\hat{\bf{\Sigma}}_{\timeIndx-1}\cdot\hat\Fmat_{\timeIndx}^T+\Qmat,
\end{equation}
\begin{equation}
\label{eqn:obs_covariance_computaion}
{\hat\Smat}_{\timeIndx|\timeIndx-1}=\hat\Hmat_{\timeIndx}\cdot{\hat{\bf{\Sigma}}}_{\timeIndx|\timeIndx-1}\cdot\hat\Hmat_{\timeIndx}^T+\Rmat.
\end{equation}
\end{subequations}
It is assumed that both ${\hat{\bf{\Sigma}}}_{\timeIndx|\timeIndx-1}$ and $\hat\Smat_{\timeIndx|\timeIndx-1}$ are nonsingular matrices. 
The matrices $\hat\Fmat_{\timeIndx}$ and $\hat\Hmat_{\timeIndx}$ are the linearized approximations of $\fvec_{\timeIndx}(\cdot)$ and $\hvec_{\timeIndx}(\cdot)$, respectively, obtained using their Jacobian matrices  evaluated at $\hat\xvec_{\timeIndx-1}$ and $\hat\xvec_{\timeIndx|\timeIndx-1}$ (see \cite[Ch. 10]{durbin2012time}), i.e., 
\begin{align}
\label{eqn:Jacobians}
\hat\Fmat_{\timeIndx}=\nabla_{\xvec}\fvec(\hat\xvec_{\timeIndx-1}); \quad 
\hat\Hmat_{\timeIndx}=
\nabla_{\xvec_{\timeIndx}}\hvec_{\timeIndx}(\hat{\xvec}_{\timeIndx|\timeIndx-1}).
\end{align} 

In the second update step, the predicted state, $ \hat{\xvec}_{\timeIndx|\timeIndx-1}$, is combined with the current observation to update the state estimate. The associated \ac{map} objective can be formulated as (see \cite[Eqn. (37)]{Welling2010}) 
\beqna
\hat{\xvec}_{\timeIndx}=\arg\min\limits_{\xvec\in \mathbb{R}^N}
 \phi_l(\xvec),
\label{EKF_update}
\eeqna
 where 
 \beqna
 \label{eq:phi_def}
 \phi_l(\xvec)\define
\|\yvec_{\timeIndx}-\hvec_{\timeIndx}(\hat{\xvec}_{\timeIndx|\timeIndx-1})-\hat{\Hmat}_{\timeIndx}(\xvec-\hat{\xvec}_{\timeIndx|\timeIndx-1})\|_{\Rmat^{-1}}\nonumber\\+\|\xvec-\hat{\xvec}_{\timeIndx|\timeIndx-1}\|_{{\hat{\bf{\Sigma}}}_{\timeIndx|\timeIndx-1}^{-1}}.
\eeqna
Solving \eqref{EKF_update} results in the updated estimator and the corresponding updated error covariance matrix as follows (see \cite[p. 5]{Welling2010}):
\begin{subequations}
    \label{eqn:updates_computation}
\begin{equation}\label{eqn:EKFUpdate}
 \hat{\xvec}_{\timeIndx} = \Kgain_{\timeIndx}\cdot \Delta \yvec_{\timeIndx} +  \hat{\xvec}_{\timeIndx|\timeIndx-1},
\end{equation}
\begin{equation}
\label{eqn:state_covariance_update}
{\hat{\bf{\Sigma}}}_{\timeIndx}=(\Imat-\Kgain_{\timeIndx}\cdot\hat\Hmat_{\timeIndx}){\hat{\bf{\Sigma}}}_{\timeIndx|\timeIndx-1}(\Imat-\Kgain_{\timeIndx}\cdot\hat\Hmat_{\timeIndx})^T+\Kgain_{\timeIndx}\cdot\Rmat\cdot\Kgain_{\timeIndx}^T,
\end{equation}
\end{subequations}
where $\Delta \yvec_\timeIndx \triangleq \yvec_{\timeIndx} -  \hat{\yvec}_{\timeIndx|\timeIndx-1}$.  The \ac{kg}, $ \Kgain_{\timeIndx}$, dictates the balance between relying on the state evolution function $\fvec_{\timeIndx}(\cdot)$ through \eqref{eqn:Pred} and the current observation $\yvec_{\timeIndx}$, and is given by 
\begin{equation}
\label{eqn:kalman_gain_computaion}
    \Kgain_{\timeIndx} = \hat{\bf{\Sigma}}_{\timeIndx|\timeIndx-1}\cdot\hat\Hmat_{\timeIndx}^T\cdot \hat{\Smat}_{\timeIndx|\timeIndx-1}^{-1}.
\end{equation} 

\vspace{-0.2cm}
\section{EKF-Aided Topology Tracking}
\label{sec:Tracking_in_the_Graph_Frequency_Domain}
 \vspace{-0.05cm}
In this section, we introduce the proposed graph topology tracking framework. We begin by introducing the considered measurement model of the topology dynamic in Subsection~\ref{ssec:Problem_Formulation}. Based on this model, we derive the corresponding \ac{ssm} formulation, and discuss its observability conditions and sparsity priors in Subsection~\ref{ssec:observability}. 
We then introduce the proposed sparsity-aware \ac{ekf} algorithm in Subsection~\ref{ssec:EKFAlg} and explore some special cases in Subsection~\ref{ssec:SpecialCases}. Finally, we discuss the algorithm's computational complexity in Subsection~\ref{ssec:Discussion}.

 \vspace{-0.2cm}
\subsection{Problem Formulation and \ac{ssm} Representation}\label{ssec:Problem_Formulation}
\subsubsection{Signal Model}
 \vspace{-0.05cm}
We consider a time-varying, undirected, weighted graph, ${\pazocal{G}}_\timeIndx({\pazocal{V}},\xi_\timeIndx)$, where $\timeIndx$ is a time index, ${\mathcal{V}}=\{1,\ldots,N\}$ denotes the stationary node set, and ${\xi_\timeIndx}=\{1,\ldots,M_\timeIndx\}$ is the set of edges with positive weights at time $\timeIndx$.  

In order to track topology changes, we 
use a \ac{gsp} model of a graph filtering process; this approach is widely used for the representation and processing of various networked signals    \cite{Sandryhaila2014,Vassilis_2016,dong2020graph}, such as power systems \cite{Rama2020Anna,GlobalSIP_Drayer_Routtenberg}, 
 water distribution networks  \cite{water_gsp}, and traffic density analysis in road networks \cite{traffic_diffusion}.
Mathematically,  at time ${\timeIndx}$, the  observed graph signal is given by
\beqna\label{linear_model}
\yvec_{{\timeIndx}}=h(\Lmat_{{\timeIndx}})\qvec_{{\timeIndx}}+\vvec_{{\timeIndx}}, 
\eeqna
where 
$\qvec_{{\timeIndx}}\in\mathbb{R}^{N}$ and  $\yvec_{{\timeIndx}}\in\mathbb{R}^{N}$ are the input and output vectors measured on the graph, and thus are  {\em{graph signals}}. 
The measurement noise,  $\vvec_{{\timeIndx}}\in\mathbb{R}^{N}$, is assumed to be temporally \ac{iid}, zero-mean Gaussian noise with covariance matrix $\Rmat$, and independent of the graph topology. 
The graph filter $h(\cdot)$ is assumed to be known and defined as in \eqref{laplacian_graph_filter}, while
 the Laplacian matrices $\Lmat_{{\timeIndx}}\in\mathbb{R}^{N\times N}$, $\forall \timeIndx$ are  unknown and generally vary over time.

 To formulate the problem of tracking a dynamic graph topology based on the measurement model in \eqref{linear_model}, we use the incidence matrix representation of the Laplacian as in \eqref{L2B}, i.e.,
\beqna\label{L2B_time_index}
\Lmat_{\timeIndx} = \Bmat\diag(\xvec_{\timeIndx}) \Bmat^T,
\eeqna
 Correspondingly, $\xvec_\timeIndx \in \mathbb{R}^\maxEdges$ is a nonnegative vector that contains the weights of all edges at time $\timeIndx$ (including the zero-weight disconnected edges), where the entry $[\xvec_\timeIndx]_m$ corresponds to the weight of edge $e_m$ at time $\timeIndx$ and is zero if $e_m \notin \xi_{\timeIndx}$,  with arbitrary mapping $\psi$ as described in Subsection \ref{ssec:Graph_Signal_Processing}. 
In the representation in \eqref{L2B_time_index}, the incidence matrix $\Bmat$ remains constant throughout the analysis, and we only need to track $\xvec_\timeIndx$.

By substituting the incidence matrix representation \eqref{L2B_time_index} at time $\timeIndx$ into \eqref{linear_model}, the measurement model can be written as follows:
\begin{subequations}
    \label{eq_model}
\beqna\label{eq_not_simplified_model}
\yvec_{{\timeIndx}}=h(\Bmat\diag(\xvec_{{\timeIndx}})\Bmat^T)\qvec_{{\timeIndx}}+\vvec_{{\timeIndx}},
\eeqna
where the edge-weight vector, $\xvec_{{\timeIndx}}$, is the state vector that evolves over time. Thus, the state vector is sparse where the zeros are associated with non-existing edges at this time, and the nonzero elements are the positive edge weights.  
Hence, our goal of tracking the topology can be formulated as tracking  $\xvec_{\timeIndx}$ based on the observed $\{\yvec_{\tau}\}_{\tau\leq {\timeIndx}}$ and knowledge of $\{\qvec_{\tau}\}_{\tau\leq {\timeIndx}}$ and the mapping $\psi$.

\subsubsection{\ac{ssm} Formulation}
We consider scenarios in which the graph topology evolves gradually over time. To capture this behavior, we model the edge-weight vector 
$\xvec_{{\timeIndx}}$ as a first-order Markov process. Together with the measurement model in \eqref{eq_not_simplified_model}, this yields a
nonlinear \ac{ssm} in which the observations are governed by the graph filter, and the underlying topology is treated as the latent state. 
Consequently, the state evolution model is described by 
\beqna  \label{eq_state_evolution}
\xvec_{{\timeIndx}} &=& \fvec_{\timeIndx}(\xvec_{{\timeIndx}-1}) + \evec_{{\timeIndx}},
\eeqna
\end{subequations}
where $\fvec_{\timeIndx}:\mathbb{R}^{\maxEdges} \mapsto {\mathbb{R}}^{\maxEdges}$,  is a known nonlinear state evolution function. In addition, $\evec_{\timeIndx}\in\mathbb{R}^{\maxEdges}$  is zero-mean, temporally \ac{iid} state-evolution noise at time $\timeIndx$, with a known covariance matrix $\Qmat$. The process noise is assumed to be statistically independent of the measurement noise $\vvec_{\timeIndx}$.

\subsubsection{Special Case -  Linear Model}
\label{sssec:LinearExample}
Consider the special case where the state evolution function $\fvec_{\timeIndx}(\xvec_{\timeIndx-1})$ and the observation function $\hvec(\xvec_{\timeIndx})$ are linear.
Thus,  $ \fvec_{\timeIndx}(\xvec_{{\timeIndx}-1})=\Fmat_{\timeIndx}\xvec_{{\timeIndx}-1}$,
and since $\hvec(\xvec_{\timeIndx})$ is a graph filter, it has to be a first-order graph filter, $h(\Lmat_{{\timeIndx}})=a_0\Imat+a_1\Bmat \diag(\xvec_{{\timeIndx}}) \Bmat^T$, where $a_0,a_1\in\mathbb{R}$  are general coefficients.
In this case, the \ac{ssm} \eqref{eq_model} can be expressed as  
\begin{subequations}
    \label{eq_linear_model}
\beqna
\xvec_{{\timeIndx}} &=& \Fmat_{\timeIndx}\xvec_{{\timeIndx}-1} + \evec_{{\timeIndx}}, \\
\label{17b}
\yvec_{{\timeIndx}} &=& \big(a_0\Imat+a_1\Bmat \diag(\xvec_{{\timeIndx}}) \Bmat^T \big)\qvec_{{\timeIndx}} + \vvec_{{\timeIndx}}.
\eeqna  
\end{subequations}
The following claim states that this special case can be referred to as a linear \ac{ssm} with the latent state $\xvec_\timeIndx$.  

\begin{Claim}\label{claim1}  
   System  \eqref{eq_linear_model} represents a linear \ac{ssm} \cite{chouzenoux2023sparse} of the form  
   \begin{subequations}
       \label{eq_model2}
   \beqna
   &&\xvec_{{\timeIndx}} =\Fmat_{\timeIndx}\xvec_{{\timeIndx}-1}+ \evec_{{\timeIndx}},\\
   &&\yvec_{{\timeIndx}}=\Hmat_{{\timeIndx}}\xvec_{{\timeIndx}}+\cvec_{\timeIndx}+\vvec_{{\timeIndx}},
   \eeqna  
   \end{subequations}
   where $\Hmat_{{\timeIndx}}\define a_1\Bmat\diag(\Bmat^T\qvec_{{\timeIndx}})$ and $\cvec_{\timeIndx}\define a_0\qvec_\timeIndx$.
\end{Claim}  

\begin{proof}  
The result follows by substituting the identity $\operatorname{diag}(\avec)\bvec=\operatorname{diag}(\bvec)\avec$ for any two vectors $\avec,\bvec\in\mathbb{R}^N$ into \eqref{17b}, with $\avec=\xvec_{{\timeIndx}}$ and $\bvec=\Bmat^T\qvec_{{\timeIndx}}$.  
\end{proof}  

In this paper,  we assume a nonlinear formulation as in \eqref{eq_model}. Nonetheless, in certain analyses, such as those presented in the following subsection, we focus specifically on the linear case.  

 \vspace{-0.2cm}
\subsection{Observability Analysis}
\label{ssec:observability}
\vspace{-0.05cm}
In this subsection, we analyze the observability of the linear system from \eqref{eq_model2}. Thus, we aim to determine under what conditions the state vector $\xvec_\timeIndx$ can be uniquely recovered from the output over a time horizon of length $T$.
We demonstrate that even in this simple linear setting, observability requires a large number of observations, which limits the
ability to detect rapid changes in the graph topology. This motivates the need for additional structural assumptions, such as sparsity.

The following theorem presents the standard observability condition for time-varying discrete-time linear systems (see \cite[Ch. 24.3]{dahleh2004lectures}) adapted for the system in \eqref{eq_model2}.
A state $\xvec_{\timeIndx}$ is $T$-step observable if it can be uniquely recovered from the output sequence $\{\yvec_k\}_{k=\timeIndx}^{\timeIndx+T-1}$ in the noiseless case of $\evec_{\timeIndx} = \mathbf{0}$ and $\vvec_{\timeIndx} = \mathbf{0}$ for all $\timeIndx$.

\begin{theorem}
\label{observbility}
Consider the time-varying linear system given in \eqref{eq_model2}.
A sufficient and necessary condition for $T$-step observability in a discrete time system from \eqref{eq_linear_model} is that the observability matrix 
\beqna
\label{matrix_O}
\vspace{-0.15cm}
\mathbf{O}_{\timeIndx}^{(T)} =
\begin{bmatrix}
\mathbf{H}_{{\timeIndx}} \\
\mathbf{H}_{{\timeIndx}+1}\Fmat_{{\timeIndx+1}} \\
\vdots \\
\mathbf{H}_{{\timeIndx+T-1}}\Fmat_{{\timeIndx}+T-1}\cdots \Fmat_{{\timeIndx}+1}
\end{bmatrix}
\vspace{-0.15cm}
\eeqna
has full rank such that $ \text{rank}(\mathbf{O}_{\timeIndx}^{(T)}) = \maxEdges $. 
\end{theorem}
\begin{proof}
The proof can be found in Appendix \ref{app:proof_observability}.
\end{proof}

Theorem~\ref{observbility} shows that the observability of the linear \ac{ssm} depends on the rank of the time-varying matrices $\Fmat_{\timeIndx}$ and $\Hmat_{\timeIndx}$, and on the observation horizon length $T$. The following corollary provides a necessary condition on the latter.
\begin{corollary}
\label{cor:Obslimit}
A necessary condition for the observability of the \ac{ssm} \eqref{eq_model2} is that the number of observations satisfies $T \geq 
\frac{N-1}{2}$.
\end{corollary}
\begin{proof}
    The result follows  from the fact that the matrix $\mathbf{O}_{\timeIndx}^{(T)}$  in \eqref{matrix_O} has dimensions
    $TN\times\maxEdges$. Thus, for $\mathbf{O}_{\timeIndx}^{(T)}$ to have full column rank, it is necessary that $TN \geq \maxEdges$.
\end{proof}

Corollary~\ref{cor:Obslimit} reveals a fundamental limitation:  the required number of observations grows quadratically with the number of nodes. This imposes a bound on how rapidly topology changes can be detected using standard \ac{ssm} techniques. In particular, at least $\left\lceil\frac{N-1}{2}\right\rceil$ measurements are required to ensure observability.
This limit stems from the fact that the \ac{ssm} treats the state $\xvec_{\timeIndx}$, i.e., the graph topology, as an arbitrary unstructured vector. However, many real-world graphs are typically sparse, with only a small number of active edges. Thus, the state vector $\xvec_{\timeIndx}$ is typically sparse, with only a few nonzero entries corresponding to active connections, i.e., \( |\xi_\timeIndx| \ll \maxEdges \).
Under such sparsity, full-rank observability of $\mathbf{O}_{\timeIndx}^{(T)}$ is not required for exact state recovery. Instead, compressed sensing techniques can be employed to recover $\xvec_{\timeIndx}$ even when $\mathbf{O}_{\timeIndx}^{(T)}$ is rank-deficient \cite{elad2010sparse}.

This insight extends to the nonlinear \ac{ssm} described in \eqref{eq_model},  where incorporating sparsity constraints can significantly reduce the number of required observations to track the state. 
This can be achieved by using a constraint of the form 
$\|\xvec_{\timeIndx}\|_0\leq K$ for some constant $K>0$ and employing sparsity-aware algorithms.

\vspace{-0.2cm}
\subsection{EKF-Based Topology Tracking}
\label{ssec:EKFAlg}
\vspace{-0.05cm}
\subsubsection{High-Level Algorithmic Framework}
We leverage the nonlinear \ac{ssm} representation of the graph topology dynamics from Subsection~\ref{ssec:Problem_Formulation} and exploit the sparsity of the state vector $\xvec_{\timeIndx}$ to design a topology tracking algorithm. In particular, we extend the \ac{ekf} to accommodate our graph-based signal model 
 while incorporating sparsity-promoting mechanisms.

As presented in Subsection \ref{ssec:EKF}, the \ac{ekf} consists of two steps: prediction and update. 
Since the state evolution in our model \eqref{eq_model} follows the standard nonlinear form given in \eqref{transition_equation},   the prediction step 
 proceeds as in the conventional \ac{ekf}. Thus,
the predicted state mean and covariance, 
$\hat{\xvec}_{\timeIndx|\timeIndx-1}$ and $\Pmat_{\timeIndx|\timeIndx-1}$, are computed using the standard \ac{ekf} equations given in \eqref{eqn:Pred} and \eqref{eqn:covariance_computation}, respectively.

The key modification to the conventional \ac{ekf} lies in the update step,  which accounts for the graph-based formulation and the sparsity of
the state vector.
For brevity, we introduce the notation 
\beqna\label{vectorial_h_graph_filter}
\hvec_{\timeIndx}(\xvec_{\timeIndx})\define h(\Bmat \diag(\xvec_{\timeIndx}) \Bmat^T) \qvec_{{\timeIndx}},\eeqna 
where $\hvec_{\timeIndx}:\mathbb{R}^{\maxEdges} \rightarrow \mathbb{R}^N$ represents the graph filtering operation.
Using the measurements model in \eqref{eq_model}, the update step is formulated as a regularized optimization problem, based on \eqref{EKF_update} with a sparsity-inducing prior on the edge-weight vector $\xvec_{\timeIndx}$. 
Following conventional convex relaxations of the $\ell_0$-(quasi) norm optimization~\cite{Polyak2013}, the sparsity-regularized optimization problem is formulated
  as follows:
 \beqna\label{ekf_update_l1_pgd}
 \hat{\xvec}_{\timeIndx}=\arg\min\limits_{\xvec\in \mathbb{R}^N}
 \phi_l(\xvec)+\mu\|\xvec\|_1,
 \eeqna
where $ \phi_l(\cdot)$ is defined in \eqref{eq:phi_def}, 
and $\mu$ is a regularization parameter that controls the sparsity level. Next, to solve \eqref{ekf_update_l1_pgd}, we describe the optimization procedure and the computation of \eqref{ekf_update_l1_pgd}, and then we elaborate on the local linearization term $\hat{\Hmat}_{\timeIndx}=\nabla_{\xvec}\hvec_{\timeIndx}(\hat{\xvec}_{\timeIndx|\timeIndx-1})$, which is required for the computation of $\phi_l(\cdot)$.

 \subsubsection{Solving the Sparsity-Regularized \ac{ekf} Update in \eqref{ekf_update_l1_pgd}}
The optimization problem in \eqref{ekf_update_l1_pgd} is convex, but nonsmooth due to the $\ell_1$-norm regularizer. Thus, while it lacks a closed-form solution \cite{hastie2015statistical}, its separable form enables efficient solvers via proximal methods. Accordingly, we employ the \ac{pgd} algorithm~\cite[p. 148]{parikh2014proximal} to iteratively solve \eqref{ekf_update_l1_pgd}. 
At each iteration,  the \ac{pgd} algorithm computes the gradient of the smooth term 
 $\phi_l(\xvec)$, $ \nabla_\xvec\phi_l(\xvec)$, and then applies the proximal mapping associated with the $\ell_1$-norm regularizer to the gradient step update. Recalling that the proximal mapping for the $\ell_1$ norm can be computed using the soft thresholding operator~\cite[Ch. 7.1]{parikh2014proximal}, 
\begin{equation}
    T_{\beta}(x) \triangleq \text{sign}(x) \cdot \max(0, |x| - \beta),
\end{equation}
 we obtain that the $m$th iteration of \ac{pgd} algorithm is given by
\begin{equation}\label{proxy_solution}
\hat{\xvec}^{(m+1)}_{\timeIndx}=
T_{\mu \rho^{(m)}} \left(\hat{\xvec}^{(m)}_{\timeIndx}-\rho^{(m)}\nabla_{\xvec}\phi_{\timeIndx}(\hat{\xvec}^{(m)}_{\timeIndx})\right).
\end{equation}
In \eqref{proxy_solution}, $\rho^{(m)}>0$ denotes the step size at the $m$th iteration.

When $\phi_l(\cdot)$ is an $\ell_2$ norm loss, the iteration in \eqref{proxy_solution} reduces to the well-known  \ac{ista}~\cite{ista2003}. In our setting, however, $\phi_l(\cdot)$ is  a sum of weighted $\ell_2$ norms (see \eqref{eq:phi_def}), whose gradient is
 \beqna\label{gradient_phi1}
 \nabla_\xvec\phi_l(\xvec)= 2
\left(\hat{\Hmat}_{\timeIndx}^T\Rmat^{-1}\hat{\Hmat}_{\timeIndx}+\hat{\bf{\Sigma}}_{\timeIndx|\timeIndx-1}^{-1}\right)(\xvec-\hat{\xvec}_{\timeIndx|\timeIndx-1})\nonumber\\-2\hat{\Hmat}_{\timeIndx}^T\Rmat^{-1}(\yvec_{\timeIndx}-\hvec_{\timeIndx}(\hat{\xvec}_{\timeIndx|\timeIndx-1})).
\eeqna
Despite the weighted form, the resulting iterative procedure is similar to the \ac{ista} in structure. Thus, we refer to the proposed sparsity-aware \ac{ekf} update rule via \eqref{proxy_solution}-\eqref{gradient_phi1} as {\em \ac{ista}-Update}.  This procedure is detailed in Algorithm~\ref{alg:ISTA-EKF}, where we initialize the iterations using the unregularized \ac{ekf} estimate from \eqref{eqn:updates_computation}.

\begingroup
\setlength{\textfloatsep}{5pt}
\begin{algorithm}
    \caption{ISTA-Update at time $\timeIndx$}
    \label{alg:ISTA-EKF} 
    \SetKwInOut{Input}{Input}
    \SetKw{KwRet}{Return}
    \Input{
    \begin{enumerate}[leftmargin=15pt]
    \item Iterations number $M$, step sizes $\{\rho^{(m)}\}_{m=0}^{M-1}$, coefficient $\mu$.
    \item Prediction estimator and 
    covariance: $\hat{\xvec}_{\timeIndx|\timeIndx-1}$, $\hat{\bf{\Sigma}}_{\timeIndx|\timeIndx-1}$.
    \item Functions $\fvec_{\timeIndx}(\cdot)$ and $h(\Bmat\operatorname{diag}(\cdot)\Bmat^T)$.
    \item Measurements $\yvec_{\timeIndx},\qvec_{\timeIndx}$.
    \item The Jacobian $\hat{\Hmat}_{\timeIndx}$ and \ac{ssm} parameters $\Qmat$ and $\Rmat$.
\end{enumerate}
      }  
        Set $\xvec_{\timeIndx}^{(0)}$ using \eqref{EKF_update}\;
        \For{$m=0,1,\ldots, M-1$}
        {%
        \begin{enumerate}[leftmargin=15pt]
            \item 

     Substitute $\hat{\xvec}^{(m)}_{\timeIndx}$ in \eqref{gradient_phi1} to obtain 
$\nabla_{\xvec}\phi_{\timeIndx}(\hat{\xvec}^{(m)}_{\timeIndx})$.

\item Calculate $
\hat{\xvec}^{(m+1)}_{\timeIndx}=
T_{\mu \rho^{(m)}} \left(\hat{\xvec}^{(m)}_{\timeIndx} \!-\!\rho^{(m)}\nabla_{\xvec}\phi_{\timeIndx}(\hat{\xvec}^{(m)}_{\timeIndx})\right)$.
        \end{enumerate}
}
\textbf{Output} 
 {$\hat{\xvec}^{(M)}_{\timeIndx}$}
\end{algorithm}
\endgroup
\subsubsection{Observations Jacobian Computation}
The \ac{ekf} update step requires the gradient of $\hvec_{\timeIndx}(\xvec)$ from \eqref{vectorial_h_graph_filter} \ac{wrt} $\xvec_{\timeIndx}$. While in general, the derivation of the Jacobian of a high-dimensional nonlinear function can be computationally intensive, in our setting, the measurement function in \eqref{vectorial_h_graph_filter} is a structured graph filter,  $h(\cdot)$,  which enables a more efficient Jacobian computation, as detailed below.

{\bf Jacobian of Graph Filter:}
The Jacobian of the graph filter \ac{wrt} edge weights quantifies how small changes in the graph topology affect the filtered signal. Due to the structure of graph filters and the linear dependence of the Laplacian on the edge weights, each partial derivative can be expressed using powers of the Laplacian, which enables an efficient recursive computation.

We utilize the polynomial representation of graph filters \cite{Sandryhaila2014}, which allows us to express the observation model  \eqref{vectorial_h_graph_filter} as 
\begin{equation}
\label{eqn:GraphFiltRep}
    \hvec_{\timeIndx}(\xvec_{\timeIndx}) = 
   h(\Lmat_{\timeIndx}) \qvec_l 
   = \sum_{p=0}^P a_p \Lmat_{\timeIndx}^p \qvec_{\timeIndx}, 
\end{equation}
where $P$ is the filter order, $\{a_p\}_{p=0}^P$ are the stationary filter coefficients, and the relation between  $\xvec_{\timeIndx}$ and $\Lmat_{\timeIndx}$ is defined in \eqref{L2B_time_index}.
Using this representation, the closed-form expression of the Jacobian is stated in the following lemma.
\begin{Lemma}\label{lemma_jacobian}  
The 
$m$th column of the Jacobian matrix  
$\nabla_{\xvec_{\timeIndx}}\hvec_{\timeIndx}(\xvec_{\timeIndx})$ for the representation in \eqref{eqn:GraphFiltRep} is given by  
\begin{equation}  
\label{jacobain_poly}  
[\nabla_{\xvec_{\timeIndx}}\hvec_{\timeIndx}(\xvec_{\timeIndx})]_{:,m} =
\sum_{p=1}^P a_p \sum_{k=0}^{p-1} {\Lmat}_{\timeIndx} ^k 
(\Bmat_{:,m} \Bmat_{:,m}^T)
{\Lmat}_{\timeIndx}^{p-k-1}\qvec_{\timeIndx},
\end{equation}  
where ${\Lmat}_{\timeIndx}$ is defined in \eqref{L2B_time_index}, and $\Bmat_{:,m}\in\mathbb{R}^{N}$ 
denotes the $m$th column of the incidence matrix $\Bmat$.
\end{Lemma}  
\begin{proof} 
 The proof appears in Appendix \ref{app;jacobian_calc}.
\end{proof}

Direct computation of the Jacobian matrix in \eqref{jacobain_poly} is computationally intensive due to repeated matrix multiplications involving powers of the Laplacian matrix. Specifically,  computing the vector ${\Lmat}_{\timeIndx} \qvec_{\timeIndx}$ requires $2N^2-N$ operations (multiplications and additions). Consequently, evaluating each term of the form ${\Lmat}_{\timeIndx}^k (\Bmat_{:,m} \Bmat_{:,m}^T) {\Lmat}_{\timeIndx}^{p-k-1} \qvec_{\timeIndx}$ 
by expressing the term as a scalar-vector product, $
(\Bmat_{:,m}^T \Lmat_{\timeIndx}^{p-k-1} \qvec_{\timeIndx} ) \cdot ( \Lmat_{\timeIndx}^{k} \Bmat_{:,m} )$, 
  requires approximately $2(p-1)N^2$ operations. 
For each column $m$, the inner sum in \eqref{jacobain_poly} contains $p$ terms,
leading to a total of $2(p^2-p)N^2$ operations per column.  
Then, summing over all degrees $p = 1$ to $P$ (outer sum in \eqref{jacobain_poly}) results in 
\beqna  
\sum_{p=1}^P 2(p^2-p) N^2 = \frac{2P(P+1)(P-1)}{3} N^2 
\eeqna  
operations,
where we used formulas for the sum of an arithmetic sequence and for the sum of squares of the first \(P\) integers \cite{pyramid_formula}.
Since the Jacobian has $N(N-1)/2$ columns (one for each possible edge), the overall computational complexity of the naive approach is $O(P^3N^4)$, which becomes computationally prohibitive for a large filter order $P$ and graph size $N$.

{\bf Efficient Computation via Dynamic Programming:}
To reduce the computational complexity of Jacobian evaluation, we propose an efficient dynamic-programming method that reuses the intermediate matrix-vector products shared across polynomial terms by reorganization of \eqref{jacobain_poly}, as follows.
\begin{Lemma}\label{Lemma_jacobian_reindex}
    The $m$th column of the Jacobian matrix of $ \hvec_{\timeIndx}(\xvec_{\timeIndx})$ 
    in \eqref{jacobain_poly} can be written as 
    \begin{equation}  
\label{jacobain_poly_reindexed}  
[\nabla_{\xvec_{\timeIndx}}\hvec_{\timeIndx}(\xvec_{\timeIndx})]_{:,m}=
\sum\nolimits_{p=0}^{P-1} 
\Dmat_p
(\Bmat_{:,m} \Bmat_{:,m}^T)
\cvec_p,
\end{equation}
where
$\cvec_p\define {\Lmat}_{\timeIndx}^{p} 
\qvec_{\timeIndx}$ and $\Dmat_p\define \sum_{r=0}^{P-1-p} 
a_{p+r+1} 
{\Lmat}_{\timeIndx}^{r}$
 for each $p=0,\ldots,P-1$.
\end{Lemma}
\begin{proof}
    The proof appears in Appendix \ref{app;jacobian_reindex}. 
\end{proof}

Based on Lemma~\ref{Lemma_jacobian_reindex},  the Jacobian matrix in \eqref{jacobain_poly_reindexed} can be
efficiently computed using the recursive updates of  $\cvec_p$ and $\Dmat_p$:
\begin{align}
&\cvec_p={\Lmat}_{\timeIndx}\cvec_{p-1},
&\cvec_0=\qvec_{\timeIndx},
\\
&\Dmat_{p}=a_{p+1}\Imat+{\Lmat}_{\timeIndx}\Dmat_{p+1}, & \Dmat_{P-1}=a_P\Imat.
\end{align}
These recursive formulations avoid repeated matrix powers and align with Horner's scheme for polynomial evaluation, which is known to be computationally optimal~\cite{Pan66}.

In addition, we note that for the incidence matrix $\Bmat$ defined in \eqref{Bmat_def}, each column has exactly two nonzero entries corresponding to the relevant edge. Therefore, the $p$-component in the sum in \eqref{jacobain_poly_reindexed} can be efficiently calculated as scalar-vector  multiplication 
\[([\cvec_p]_n-[\cvec_p]_k)\cdot([\Dmat_p]_{:,n}-[\Dmat_p]_{:,k}),\] where $n,k$ are the nonzero elements in $\Bmat_{:,m}$ associated with the edge $e_m=(n,k)$, based on the fixed mapping of nodes to edges, $\psi$. 
The final algorithm is summarized in Algorithm~\ref{alg:DP_Jacobian}. 

\begingroup
\setlength{\textfloatsep}{2pt}
\begin{algorithm}
\caption{Dynamic programming for efficient Jacobian computation}
\label{alg:DP_Jacobian}

\textbf{Input:} 
${\xvec}_{\timeIndx} $, mapping $\psi$, 
${\Lmat}_{\timeIndx}$,$\qvec_{\timeIndx} $,  $\{a_p\}_{p=0}^P.$
\begin{enumerate}[leftmargin=15pt]
\item \textbf{Precompute $\{\cvec_p,\Dmat_p\}_{p=0}^{P-1}$:}\begin{enumerate}
    \item  $\cvec_0 \gets \qvec_{\timeIndx}$ 
    \item $\Dmat_{P-1} \gets a_{P}\Imat$ 
    \item \textbf{For $p = 1$ to $P-1$:}\begin{enumerate}
    \item  $\cvec_p \gets {\Lmat}_{\timeIndx} \cvec_{p-1}$
    \item  $\Dmat_{P-1-p} \gets a_{P-p}\Imat+{\Lmat}_{\timeIndx} \Dmat_{P-p}$
            \end{enumerate}
\end{enumerate}

\item \textbf{Compute each column of $\nabla_{\xvec_{\timeIndx}}\hvec_{\timeIndx}(\xvec_{\timeIndx})$:}\\
 \textbf{For $m = 1$ to $\maxEdges$:}\\
$(n,k)\gets\psi^{-1}(m)$\\
    $[\nabla_{\xvec_{\timeIndx}}\hvec_{\timeIndx}(\xvec_{\timeIndx})]_{:,m}=([\cvec_p]_{n}-[\cvec_p]_{k})\cdot([\Dmat_p]_{:,n}-[\Dmat_p]_{:,k})$

\end{enumerate}
\textbf{Output} $\nabla_{\xvec_{\timeIndx}}\hvec_{\timeIndx}(\xvec_{\timeIndx})$
\end{algorithm}
\endgroup
The computational complexity of Step 1) in Algorithm \ref{alg:DP_Jacobian} is 
$O(PN^3)$, since it involves approximately $2(P-1)N^2+2(P-2)N^3$ operations. 
Step 2) performs $\maxEdges$ iterations of $O(PN)$ 
operations, 
resulting in a complexity of $O(PN^3)$. 
Overall, the total computational complexity is $O(P N^3)$, which is a significant improvement over the $O(P^3 N^4)$ cost of the direct  Jacobian evaluation.

 \subsubsection{Algorithm Summary}
The proposed \ac{gsp}-\ac{ekf} algorithm extends the classical \ac{ekf} to track sparse, time-varying graph topologies by incorporating sparsity priors into the state update, and leveraging the structure of polynomial graph filters. 
To reduce complexity, the Jacobian matrix ${\Hmat}_{\timeIndx}(\hat{\xvec}_{\timeIndx|\timeIndx-1})$ is efficiently computed using the recursive scheme in Algorithm~\ref{alg:DP_Jacobian}, derived from the polynomial filter representation in \eqref{jacobain_poly}.
The state update is formulated as a convex sparsity-regularized problem \eqref{ekf_update_l1_pgd} and solved via the \ac{pgd} algorithm 
(Algorithm \ref{alg:ISTA-EKF}).  Within the \ac{ekf} loop, the predicted covariance from \eqref{eqn:state_covariance_computaion} is propagated forward as an approximation of the posterior covariance, while the predicted $\hat{\xvec}_{\timeIndx}$ in \eqref{eqn:EKFUpdate} serves as the initialization   $\xvec_{\timeIndx}^{(0)}$  for the iterative update. The full procedure is summarized in Algorithm~\ref{alg:GSP-EKF}.
\vspace{-0.25cm}
\begingroup
\setlength{\textfloatsep}{5pt}
\begin{algorithm}
    \caption{\ac{gsp}-\ac{ekf} at time $\timeIndx$}
    \label{alg:GSP-EKF} 
    \SetKwInOut{Input}{Input}
    \SetKw{KwRet}{Return}
    \Input{
     \begin{enumerate}[leftmargin=15pt]
    \item Previous-iteration estimator and its covariance: $\hat{\xvec}_{\timeIndx-1}$, $\hat{\bf{\Sigma}}_{\timeIndx-1}$.
    \item Functions $\fvec_{\timeIndx}(\cdot)$ and $h(\Bmat\operatorname{diag}(\cdot)\Bmat^T)$.
    \item Measurements $\yvec_{\timeIndx},\qvec_{\timeIndx}$.
    \item  \ac{ssm} parameters $\Qmat$ and $\Rmat$.
     \item \ac{ista} parameters $(M,\mu,\{\rho^{(m)}\}_{m=1}^M)$.
\end{enumerate}}  
        {%
            \underline{Prediction step:}
            \begin{align}
            \label{prediction_step}
                \hat{\xvec}_{\timeIndx|\timeIndx-1} &=  \fvec_{\timeIndx}(\hat{\xvec}_{\timeIndx-1})\\
                \label{prediction_obs_step1}
                \hat{\yvec}_{\timeIndx|\timeIndx-1} &=  \hvec_{\timeIndx}(\hat{\xvec}_{\timeIndx|\timeIndx-1})\end{align}
            Compute  $\hat{\Hmat}_{\timeIndx}=\nabla_{\xvec_{\timeIndx}}\hvec_{\timeIndx}(\hat{\xvec}_{\timeIndx|\timeIndx-1})$ via  Algorithm \ref{alg:DP_Jacobian};
            \begin{align}\label{prediction_step2}\hat{\bf{\Sigma}}_{\timeIndx|\timeIndx-1} &= \hat{\Fmat}_{\timeIndx}\hat{\bf{\Sigma}}_{\timeIndx-1}\hat{\Fmat}_{\timeIndx}^T + {\Qmat} \\
            \label{prediction_step3}
{\hat\Smat}_{\timeIndx|\timeIndx-1}&=\hat{\Hmat}_{\timeIndx}\cdot{\hat{\bf{\Sigma}}}_{\timeIndx|\timeIndx-1}\cdot\hat{\Hmat}_{\timeIndx}^T+\Rmat
            \end{align}
            \underline{Update step:}
              \begin{align}
            \label{kg_alg}
            \Kgain_{\timeIndx} &= \hat{\bf{\Sigma}}_{\timeIndx|\timeIndx-1}\cdot\hat{\Hmat}_{\timeIndx}^T\cdot \hat{\Smat}_{\timeIndx|\timeIndx-1}^{-1}\end{align}
            Solve \eqref{ekf_update_l1_pgd} for $\hat{\xvec}_{\timeIndx}$, by using Algorithm \ref{alg:ISTA-EKF};
            \begin{align}
            \label{update_step_freq3}
                \hat{\bf{\Sigma}}_{\timeIndx} &=(\Imat-\Kgain_{\timeIndx}\cdot\hat{\Hmat}_{\timeIndx}){\bf{\Sigma}}_{\timeIndx|\timeIndx-1}(\Imat-\Kgain_{\timeIndx}\cdot\hat{\Hmat}_{\timeIndx})^T\nonumber\\&\qquad + \Kgain_{\timeIndx}\cdot\Rmat\cdot\Kgain_{\timeIndx}^T
                \end{align}       
            }
        \textbf{Output} {$\hat{\xvec}_{\timeIndx}$}
\end{algorithm}
\endgroup

\vspace{-0.75cm}
\subsection{Special Cases}
\label{ssec:SpecialCases}
In this subsection, we examine two special cases: $(i)$  the linear measurement model from Subsection \ref{sssec:LinearExample}; and $(ii)$ known graph support.

\subsubsection{Linear Case}
The following claim shows that under specific initialization and under the linear \ac{ssm} described in Example \ref{sssec:LinearExample}, the proposed sparsity-aware  \ac{gsp}-\ac{ekf} method reduces to our earlier method  in~\cite{dabush2024Routtenberg_Kalman}. However, the current framework generalizes to nonlinear models and general initialization strategies.
\begin{Claim}\label{claim;starting_point}
Let the initial state,  $\xvec_{\timeIndx}^{(0)}$, of \ac{ista}-Update in Algorithm~\ref{alg:ISTA-EKF}, be set to the unregularized \ac{ekf} estimator, which implies  
$\nabla_{\xvec}\phi_l(\hat{\xvec}_{\timeIndx}^{(0)})=\zerovec$. Thus, the first iteration of \eqref{proxy_solution} is simplified to 
\begin{equation}\label{proxy_solution_1_iter}
\hat{\xvec}^{(1)}_{\timeIndx}=
T_{\mu\rho^{(0)}} \left(\hat{\xvec}_{\timeIndx}^{(0)}\right). 
\end{equation}
In particular, for a linear \ac{ssm}, \eqref{proxy_solution_1_iter} coincides with the second approach of the graph tracking \ac{kf} (GT-KF 2)  in \cite{dabush2024Routtenberg_Kalman}.
\end{Claim}
\begin{proof}
    The proof appears in Appendix \ref{app;proof_initialization}. 
\end{proof}

\subsubsection{Known Support and Unknown Weights}\label{subsub;oracle}
In the following, we consider a scenario where the graph connectivity (i.e., support) is known, while the edge weights evolve over time and must be estimated.
Representative examples include power grids with fixed topology and varying line impedances; 
traffic networks with fixed roads and time-varying congestion levels; and social networks with stable relationships between individuals and evolving interaction strengths. In such scenarios, the estimation problem reduces to tracking time-varying edge weights over a fixed support.

The known graph support can be enforced by using a masking matrix $\Mmat_{\timeIndx}\in\mathbb{R}^{\maxEdges\times\maxEdges}$, which preserves the rows corresponding to the active edges at time \( \timeIndx \). That is, \( [\Mmat_\timeIndx]_{ii} = 1 \) if \( i \in \xi_{\timeIndx} \), and \( [\Mmat_{\timeIndx}]_{ii} = 0 \) otherwise. This results in the following modification of \eqref{eq_state_evolution} as follows:
\begin{equation}
\label{eq_state_evolution_known_support}
\xvec_{{\timeIndx}} = \Mmat_{\timeIndx}\cdot\big(\fvec_{\timeIndx}(\xvec_{{\timeIndx}-1}) + \evec_{{\timeIndx}}\big),
\end{equation}
 where we assume that $\evec_{\timeIndx} \sim \mathcal{N}(\bar{\evec}_{\timeIndx}, \Qmat)$, where $\bar{\evec}_{\timeIndx}$ accounts for potential jumps in edge weights due to newly established connections.
For the state evolution model in \eqref{eq_state_evolution_known_support}, the prediction step can be written as:
\begin{equation}
\label{EKF_prediction_known_connections2}
\hat{\xvec}_{\timeIndx | \timeIndx-1} = \Mmat_{\timeIndx}\cdot\big(
\fvec_{\timeIndx}(\hat{\xvec}_{\timeIndx-1})+\bar{\evec}_{\timeIndx}\big).
\end{equation}
Accordingly, the state evolution noise covariance and the Jacobian are given by $\Mmat_{\timeIndx}\Qmat\Mmat_{\timeIndx}$ and $\Mmat_{\timeIndx}\Fmat_{\timeIndx}$, respectively, where $\Qmat$ and $\Fmat_{\timeIndx}$ correspond to the unmasked model. This yields the prediction covariance:
\begin{equation}
\label{eqn:state_covariance_computaion_masked}
{\hat{\bf{\Sigma}}}_{\timeIndx|\timeIndx-1}=\Mmat_{\timeIndx}\hat\Fmat_{\timeIndx}\hat{\bf{\Sigma}}_{\timeIndx-1}\hat\Fmat_{\timeIndx}^T\Mmat_{\timeIndx}+\Mmat_{\timeIndx}\Qmat\Mmat_{\timeIndx}.
\end{equation}

The known graph support can also be enforced as a hard constraint into the update step of the \ac{gsp}-\ac{ekf} from \eqref{EKF_update} by reformulating it as a constrained optimization problem:
\beqna
\begin{array}{lr}
\hat{\xvec}_{\timeIndx}=\arg\min\limits_{\xvec\in \mathbb{R}^N}
\phi_l(\xvec)
\\
{\text{subject to }}\xvec_{\bar{\xi}_{\timeIndx}}=\zerovec,
\end{array} 
\label{EKF_update_known_connections}
\eeqna
where $\phi_l(\xvec)$ is defined in 
\eqref{eq:phi_def}, $\bar{\xi}_{\timeIndx} \define \xi \setminus \xi_{\timeIndx}$ denotes the set of inactive edges at time $\timeIndx$, and 
${\hat{\bf{\Sigma}}}_{\timeIndx|\timeIndx-1}^{-1}$ is replaced by its pseudo-inverse ${\hat{\bf{\Sigma}}}_{\timeIndx|\timeIndx-1}^{\dagger}$ to account for the reduced-rank structure corresponding to the active edges. 
The closed-form solution to \eqref{EKF_update_known_connections} yields the updated estimator:  
\begin{equation}
\label{EKF_update_known_connections2}
\left\{
\begin{array}{l}
\left[\hat{\xvec}_{\timeIndx} \right]_{\xi_{\timeIndx}} = \left[\hat{\xvec}_{\timeIndx | \timeIndx-1} +\Kgain_{\timeIndx}\cdot \Delta \yvec_{\timeIndx}\right]_{\xi_{\timeIndx}},\\
\left[\hat{\xvec}_{\timeIndx} \right]_{\bar{\xi}_{\timeIndx}} = 0.
\end{array}
\right.
\end{equation}
where $\Kgain_{\timeIndx}$ is the \ac{kg}, which is computed similarly to \eqref{eqn:kalman_gain_computaion} but incorporates the masked covariance matrix from \eqref{eqn:state_covariance_computaion_masked}  for ${\hat{\bf{\Sigma}}}_{\timeIndx|\timeIndx-1}$.

\subsection{Discussion}
\label{ssec:Discussion}
This work introduces a novel sparse nonlinear \ac{ssm} framework for joint tracking of time-varying edge sets and edge weights over graphs.
Unlike previous approaches, we model nonlinear topology-dependent dynamics via graph-filter processes, which leads to the development of \ac{gsp}-based \ac{ekf} for sparse state estimation. 
In this subsection, we discuss the computational complexity and future directions. 

\subsubsection{Computational Complexity}
In the following, we analyze the computational complexity of a single iteration of the  \ac{gsp}-\ac{ekf}.  The prediction step requires evaluating the state transition function \(\fvec_{\timeIndx}(\cdot)\)  and its Jacobian $\Fmat_{\timeIndx}$, both standard in \ac{ekf} settings. 
 In addition,  constructing the Laplacian matrix via edge-wise updates is linear in the number of edges, 
 which leads to an $O(N^2)$ computation for dense graphs.
Evaluating the measurement function as written in \eqref{eqn:GraphFiltRep} requires \(PN^2\) multiplications, computing the covariance \({\bf{\Sigma}}_{\timeIndx|\timeIndx-1}\) involves \(2(\maxEdges)^3\) multiplications, and calculating \(\Smat_{\timeIndx|\timeIndx-1}\) requires \(N(\maxEdges)^2+N^2(\maxEdges)\) multiplications. Additionally, computing the Jacobian matrix using Algorithm~\ref{alg:DP_Jacobian} incurs a cost of \(PN^3\) multiplications.

In the update step, the computation of the \ac{kg} from \eqref{kg_alg} requires \(N(\maxEdges)^2+N^2(\maxEdges)\) 
multiplications. 
The main bottleneck is
the calculation of $\hat{\xvec}_{\timeIndx}$ via Algorithm~\ref{alg:ISTA-EKF}  the inversion of $\hat{\bf{\Sigma}}_{\timeIndx|\timeIndx-1}$ in~\eqref{gradient_phi1}, with a worst-case complexity of \(O((\maxEdges)^3) \). However, if the algorithm is initialized as described in Claim~\ref{claim;starting_point} and only one ISTA iteration is used, the computation involves only \( N\maxEdges \) multiplications.
The covariance matrix $\hat{\bf{\Sigma}}_{\timeIndx}$ in \eqref{update_step_freq3} adds a further \(2N(\maxEdges)^2 + (\maxEdges)^3\) multiplications. Our simulations indicate that a single ISTA iteration is often sufficient to achieve good performance.

Overall, the dominant computational cost arises from covariance updates and the inversion of $\hat{\bf{\Sigma}}_{\timeIndx|\timeIndx-1}$, as in standard \ac{ekf} settings.
This complexity can be notably reduced using existing methods for \ac{ekf}, e.g., low-rank and diagonal approximations~\cite{buchnik2023gspkalmannet,chang2022diagonal}, as well as parallelization strategies for large-scale graphs.
In terms of our \ac{gsp} setting,  our Jacobian computation reduced the cost from $O(P^3N^4)$ to $O(PN^3)$ using Algorithm~\ref{alg:DP_Jacobian}.



 \subsubsection{Future Research Directions}
 Our proposed framework of casting dynamic graphs as \acp{ssm} gives rise to multiple possible extensions. 
 From a modeling perspective, an important direction is extending the model to handle partial measurements available from the network, allowing for more realistic data acquisition.
 Furthermore, alternative filtering techniques, such as particle filters~\cite{arulampalam2002tutorial} or variational inference~\cite{smidl2008variational} methods, could improve performance in nonlinear and non-Gaussian settings, and possibly combined with machine learning methods for balancing the excessive complexity~\cite{nuri2024learning}. Machine learning tools can potentially be integrated with our proposed sparsity-aware \ac{gsp}-\ac{ekf},  adapting to scenarios with unknown measurement functions or noise statistics using the KalmanNet methodology \cite{revach2022kalmannet} or via alternative forms of learning-aided \acp{kf}~\cite{shlezinger2024ai}. 
Finally, developing a distributed implementation, possibly leveraging tools from distributed tracking~\cite{cattivelli2010diffusion}, could enable scalability for large-scale networks, making the method more practical for decentralized applications.
\section{Numerical Study}
\label{ssec:simulations}
In this section, we numerically evaluate the proposed \ac{gsp}-\ac{ekf} algorithm and compare it with existing approaches. In Subsection~\ref{ssec:setup}, we introduce the simulation setup, while in Subsections~\ref{ssec;sim_linear}-\ref{ssec;sim_nonlinear} we present detailed results for linear and nonlinear cases, respectively.

\subsection{Experimental Setup}
\label{ssec:setup}
In our study, we compare the following estimators\footnote{The code is available  at \url{https://github.com/lital-dgold/PolyEKF.git}}: 
\begin{enumerate}
\item \textbf{\ac{gsp}-\ac{ekf}} (Algorithm \ref{alg:GSP-EKF}) with $M=1$ and $\mu\rho^{(0)}=0.25$. 
\item  \textbf{Change-det} method of \cite{10333427}, where the adjacency matrix is replaced by the Laplacian matrix (that is, the optimization parameter is $\xvec$, which composes the Laplacian by \eqref{L2B_time_index} instead of the adjacency). 
The additional parameters of this method, including the time window length \( W \) and the regularization parameters \( \lambda_1 \) and \( \lambda_2 \) (Alg. 1 in  \cite{10333427}), were hand-tuned to optimize performance for each scenario.

\item \textbf{\ac{ekf}}, i.e., Algorithm \ref{alg:GSP-EKF} where $\hat{\xvec}_{\timeIndx}$ is obtained from \eqref{eqn:EKFUpdate} and \eqref{eqn:kalman_gain_computaion}, and the computation of $\hat{\Hmat}_{\timeIndx}$ is straight forward calculation of \eqref{jacobain_poly} without Algorithm \ref{alg:DP_Jacobian}. At the end of the update step, negative entries of the state vector are set to $0$.
\item \textbf{Oracle-\ac{gsp}-\ac{ekf}} (see Subsection~\ref{subsub;oracle}), which knows the true support of the state $\xvec_\timeIndx$, $\xi_\timeIndx$, (unknown in our setting).
  \end{enumerate}

We consider a sparse graph of \(N\) nodes, initialized with \(3N\) edges chosen uniformly at random, each assigned a unit weight. According to Theorem \ref{observbility}, for a system to be $T$-step observable, one should take  $T>N$ time samples. Therefore, every \(2N\) time steps, either one edge is added with weight sampled from \(\mathcal{N}(1, 0.01)\) (so that, with high probability, the edge weight would be positive) or one existing edge is removed. This setup simulates gradual changes in the network topology over time. 

We use the \ac{ssm} \eqref{eq_model}, where the state transition function is \(\fvec_{\timeIndx}(\xvec) = \xvec\). 
The signals $\{\qvec_k\}_{k\leq\timeIndx}$ were generated from Gaussian distribution $\mathcal{N}(\zerovec, \Imat)$, 
and the noise covariances of $\evec_{\timeIndx}$ and $\vvec_{\timeIndx}$ are $\Qmat_{\timeIndx}=\sigma_{\evec}^2\operatorname{diag}(\onevec_{\xi_{\timeIndx}})$ and $\Rmat=\sigma_{\vvec}^2\Imat$, respectively, with $\sigma_{\evec}$ and $\sigma_{\vvec}$ defined in each simulation.
It is emphasized that the compared estimation algorithms which require knowledge of the \ac{ssm} are given $\Qmat_{\timeIndx}=\sigma_{\evec}^2\Imat$. 
 In addition, all methods are initialized with \(\hat{\xvec}_0 = \onevec\) and \(\hat{\Sigmamat}_0 = 0.25 \cdot \Imat\), except for Oracle-\ac{gsp}-\ac{ekf}, which is initialized with the true state and \(\hat{\Sigmamat}_0 = 0.25 \cdot \operatorname{diag}(\onevec_{\xi_0})\). This allows for a fair comparison under imperfect initialization conditions, with Oracle-\ac{gsp}-\ac{ekf} serving as an (upper) performance benchmark.

 For the observation model, we use the following settings:
 \begin{itemize}
     \item {\em Lin}: A linear model with $ \hvec_{\timeIndx}(\xvec_{\timeIndx}) = \Lmat_{\timeIndx}$;
     \item {\em NL4}: Fourth-order nonlinear model given by $\hvec_{\timeIndx}(\xvec_{\timeIndx})=\Imat+\Lmat_{\timeIndx}+\Lmat_{\timeIndx}^2+0.1\Lmat_{\timeIndx}^3+\Lmat_{\timeIndx}^4$;
     \item {\em NL5}: Fifth-order nonlinear model given by $\hvec_{\timeIndx}(\xvec_{\timeIndx})=\Imat+\Lmat_{\timeIndx}+0.8\Lmat_{\timeIndx}^2+0.6\Lmat_{\timeIndx}^3+0.4\Lmat_{\timeIndx}^4+0.2\Lmat_{\timeIndx}^5$.
     \item {\em NLP}: $P$-order nonlinear model given by \( \hvec_{\timeIndx}(\xvec_{\timeIndx}) = \sum_{p=0}^P \frac{1}{2^p} \Lmat_{\timeIndx}^p \), where \( P \), the polynomial filter order, varies from 1 to \( N - 1 \).
 \end{itemize}

Our main performance measures are 
$(i)$ normalized \ac{mse}, $\frac{1}{\maxEdges}\times$ \ac{mse}, to evaluate the weights estimation; and
$(ii)$ \ac{eier}, defined as
\begin{equation}
    \text{EIER} := \frac{1}{N(N-1)} 
     \times|(\xi_{\timeIndx}\setminus\hat{\xi}_{\timeIndx})\cup(\hat{\xi}_{\timeIndx}\setminus{\xi}_{\timeIndx})|
    \times 100\%,
\end{equation}
which measures the percentage of incorrectly estimated connections. An edge is declared present if \([\xvec_\timeIndx]_n\) exceeds a threshold of $0.1$, to prevent false detections due to small numerical fluctuations or noise in low-magnitude entries. 

\subsection{Linear Model}\label{ssec;sim_linear}
We begin by evaluating the performance of the proposed algorithm in a linear case ({\em Lin}). The graph consists of $N=20$ nodes, and the variances  of $\evec_{\timeIndx}$ and $\vvec_{\timeIndx}$ set to $\sigma_{\evec}=0.01$, and $\sigma_{\vvec}=0.01$, respectively. 
The performance was evaluated for $1,000$ Monte Carlo simulations, each of length $159$ time samples, with the per-entry \ac{mse} and the \ac{eier} vs. time are reported in Fig.~\ref{fig:linear_case_subfigs}. 

We observe Fig.~\ref{fig_mse_vs_t_lin} that error peaks occur at time indices $t = 40$, $80$, and $120$ corresponding to abrupt changes in the graph support $\xi_{\timeIndx}$. As expected, {Oracle-\ac{gsp}-\ac{ekf}}, which has prior knowledge of the support of $\xi_{\timeIndx}$, achieves the lowest \ac{mse} and recovers most rapidly after topology changes or initialization errors compared with the other estimators. It also achieves perfect support recovery (zero \ac{eier} in Fig. \ref{fig_eier_vs_t_lin}) since it uses the true support set, $\xi_\timeIndx$.
{Change-det}, whose window size is set to $W=N$, 
converges faster than the \ac{gsp}-\ac{ekf} during the first $2N$ time samples.  However, for $\timeIndx > 2N$, the \ac{gsp}-\ac{ekf} outperforms it by leveraging an adaptive weighting mechanism that emphasizes recent information,  whereas Change-det method assigns equal weight to all $N$ past samples, despite the evolving state defined in \eqref{linear_model}. The conventional \ac{ekf} performs worst overall, as it does not exploit the sparsity prior. 
Finally, it can be seen that the \ac{gsp}-\ac{ekf} gradually approaches the benchmark performance of the  Oracle-\ac{gsp}-\ac{ekf}, despite the fact that it has no prior knowledge of the support. 



\begin{figure}
    \centering
    \subfloat[Normalized \ac{mse} over time.]{
        \includegraphics[width=0.85\columnwidth]{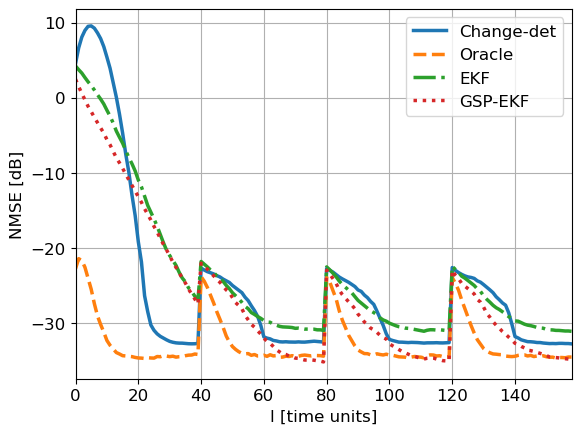}
        \label{fig_mse_vs_t_lin}
    }

    \subfloat[\ac{eier} over time.]{
        \includegraphics[width=0.85\columnwidth]{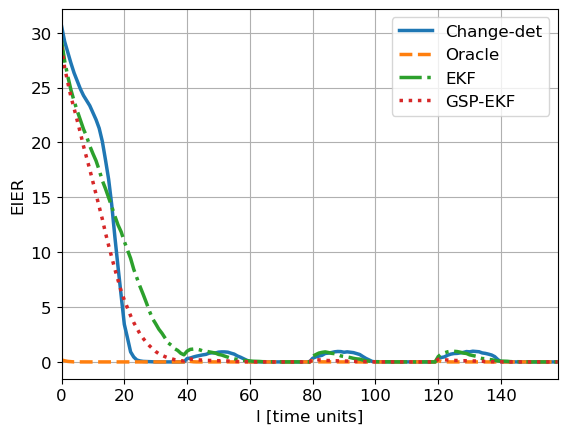}
        \label{fig_eier_vs_t_lin}
    }
    \caption{Performance vs. time, {\em Lin} \ac{ssm}.}
    \label{fig:linear_case_subfigs}
\end{figure}

\subsection{Nonlinear Model}\label{ssec;sim_nonlinear}
As our main focus is on nonlinear settings, we dedicate the main bulk of our numerical study to the evaluation of our \ac{gsp}-\ac{ekf} in nonlinear measurement models across a range of metrics and scenarios. Specifically, we evaluate how its performance evolves with time, with different noise levels, and different graph parameters, as well as its runtime.

\subsubsection{Performance vs. Time}\label{ssec;nonlinear_ver1_vs_time}
We commence by repeating the numerical study reported for the linear case, while modifying the observation model to {\em NL4}.  The resulting performance metrics are reported in Fig.~\ref{fig:nonlinear_case_subfigs}.
%
Oracle-\ac{gsp}-\ac{ekf} consistently achieves the lowest \ac{mse} and converges fastest, owing to its access to the true support $\xi_{\timeIndx}$.
The remaining methods start with high \ac{mse} due to a mismatch in the initial state $\xvec_0$ and the need to learn the full connectivity.
While Change-det initially reduces error more quickly, the \ac{gsp}-\ac{ekf} achieves lower \ac{mse} for $\timeIndx > 2N$ by adaptively weighting past samples, in contrast to the equal-weighting strategy of Change-det.
The \ac{ekf} exhibits a slower decline in error and maintains the highest \ac{mse}, as it does not exploit the sparsity prior.
Similar findings are noted in the \ac{eier} comparison in  Fig. \ref{fig_eier_vs_t_nonlinear}. There, we also note that for $\timeIndx > 4N$, the \ac{gsp}-\ac{ekf} consistently outperforms both the \ac{ekf} and Change-det. This improvement stems from two key factors: $(i)$ Change-det relies on a linear approximation of the topology measurement model, which fails to capture the nonlinear dependencies present in this scenario, and $(ii)$ it requires a batch of recent samples, which delays its adaptation to abrupt changes.

\begin{figure}
    \centering
    \subfloat[Normalized \ac{mse} over time.]{
        \includegraphics[width=0.85\columnwidth]{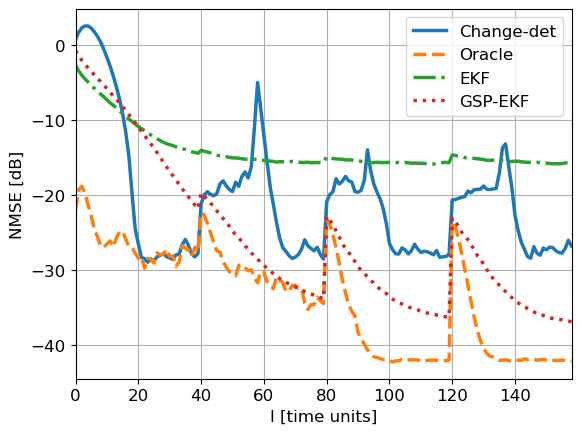}
        \label{fig_mse_vs_t_nonlin}
    }

    \subfloat[\ac{eier} over time.]{
        \includegraphics[width=0.85\columnwidth]{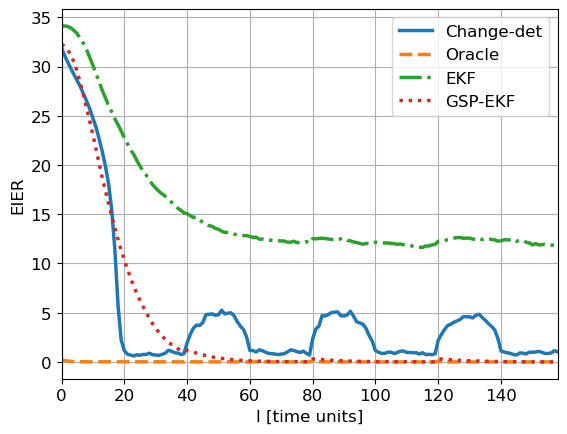}
        \label{fig_eier_vs_t_nonlinear}
    }
    
    \caption{Performance vs. time, {\em NL4} \ac{ssm}.}
    \label{fig:nonlinear_case_subfigs}
\end{figure}


We proceed to the more challenging nonlinearity of {\em NL5}. Here, the graph has $N=10$ nodes, and is initialized with $15$ edges chosen uniformly at random, each assigned a unit weight. The noise variances are $\sigma_{\evec}=0.1$,  and $\sigma_{\vvec}=\sqrt{0.2}$.  For Change-det method, the time window length is set to $W=0.5N$ time samples, and 
$\lambda_1=3.16$ and $\lambda_2=0.316$.
The performance metrics versus time are reported in Fig.~\ref{fig:nonlinear_case2_subfigs}. 
In Fig.~\ref{fig_mse_vs_t_nonlin2}, we note that the \ac{mse} of Change-det plateaus around \(-5\) dB because it is unaware of the noise variance, which is significantly more dominant in this highly nonlinear case. In addition, the process noise is also more influential, so the equal-weighting strategy of Change-det treats the fast-evolving state as constant, which increases the error.  
The gap between the \ac{gsp}-\ac{ekf} and the \ac{ekf} stems from the fact that the \ac{ekf} does not exploit the sparsity prior. However, this difference is less pronounced here because the graph size is smaller, the percentage of connected edges is higher, and the state is less sparse, which eliminates the sparsity advantage of the \ac{gsp}-\ac{ekf}.
Observing the \ac{eier} in Fig.~\ref{fig_eier_vs_t_nonlinear2}, we note that for \( \timeIndx > 2N \), the \ac{gsp}-\ac{ekf} performs best, notably outperforming  Change-det, which struggles to cope with strong noise values, dominant nonlinearities, and rapid variations. The \ac{ekf} handles noise better than Change-det in this scenario but does not use the sparsity of the graph, so its performance in estimating the connectivity is worse in the long run compared to the \ac{gsp}-\ac{ekf}.

\begin{figure}
\centering
\subfloat[Normalized \ac{mse} over time.]{
\includegraphics[width=0.85\columnwidth]{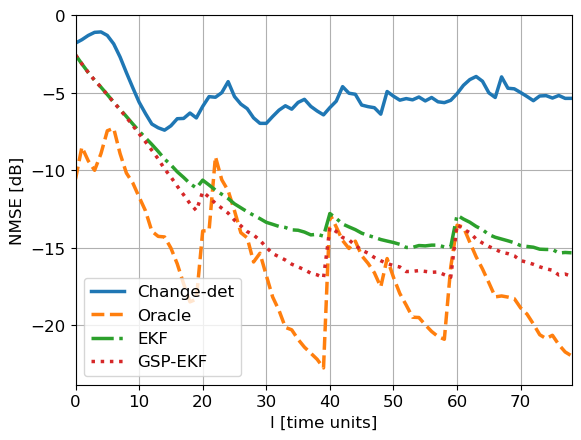}
\label{fig_mse_vs_t_nonlin2}
}

\subfloat[\ac{eier} over time.]{
    \includegraphics[width=0.85\columnwidth]{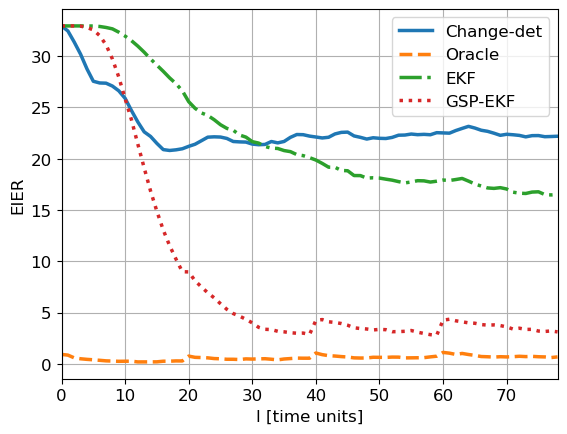}
    \label{fig_eier_vs_t_nonlinear2}
}

\caption{Performance vs. time, {\em NL5} \ac{ssm}.}
\label{fig:nonlinear_case2_subfigs}
\end{figure}

\subsubsection{Performance vs. Noise Level}\label{ssec;performance_vs_snr}
We proceed to evaluating our \ac{gsp}-\ac{ekf} under different noise levels. Unless stated otherwise, we employ the {\em NL5} \ac{ssm} with the configuration used for the study reported in Fig.~\ref{fig:nonlinear_case2_subfigs}, with the exception that the performance is averaged over $100$ graph realizations along $79$ time instances. Here, we set \( \sigma_{\evec} = \sigma_{\vvec} \), and evaluate performance vs noise level in Fig.~\ref{fig:nonlinear_noise_sensitivity}.

In Fig.~\ref{fig_mse_vs_snr_nonlin2}, we observe that for small noise variance values, Change-det, the \ac{gsp}-\ac{ekf}, and the \ac{ekf} reach a plateau, since even under weak noise, it remains challenging to estimate both the edge set and the edge weights. Notably, Change-det suffers from higher \ac{mse} even under low noise levels due to model mismatch. It relies on a linear approximation that is inaccurate in the presence of strong nonlinearity in the data.   %
%
In Fig.~\ref{fig_eier_vs_snr_nonlinear2}, the \ac{eier} of Change-det, the \ac{gsp}-\ac{ekf}, and the \ac{ekf} are shown to increase with growing noise variances, with the \ac{gsp}-\ac{ekf} notably outperforming Change-det and the \ac{ekf}. At high noise levels, the \ac{eier} of Oracle-\ac{gsp}-\ac{ekf} is also observed to be greater than zero. This degradation results from the breakdown of the Gaussian assumption in the edge weights, leading Oracle-\ac{gsp}-\ac{ekf} to occasionally estimate negative weights, which are subsequently set to zero, thereby increasing the \ac{eier}. 

\begin{figure}
\centering
\subfloat[Normalized \ac{mse} vs.  $\sigma_{\evec}$ ($\sigma_{\evec}=\sigma_{\vvec}$).]{
\includegraphics[width=0.85\columnwidth]{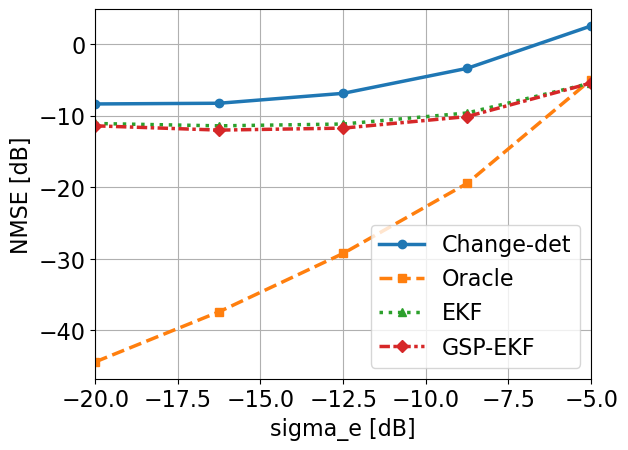}
\label{fig_mse_vs_snr_nonlin2}
}

\subfloat[\ac{eier} vs.  $\sigma_{\evec}$ ($\sigma_{\evec}=\sigma_{\vvec}$).]{
    \includegraphics[width=0.85\columnwidth]{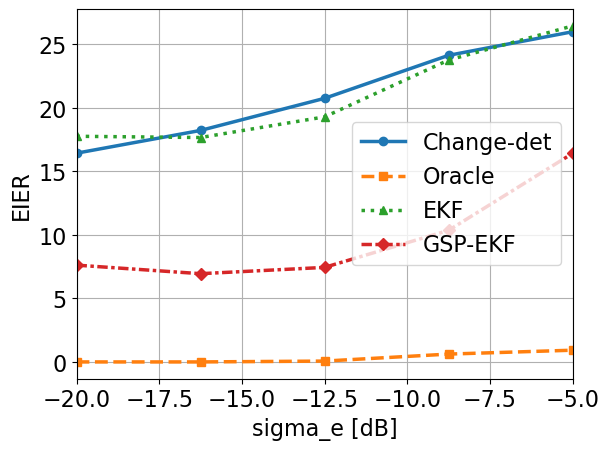}
    \label{fig_eier_vs_snr_nonlinear2}
}
\caption{Performance vs. noise level, {\em NL5} \ac{ssm}.}
\label{fig:nonlinear_noise_sensitivity}
\end{figure}

\subsubsection{Performance vs. Graph Parameters}\label{ssec;performance_vs_parameters}
We proceed to evaluating our \ac{gsp}-\ac{ekf} under different graph parameters. Unless stated otherwise, we employ the {\em NL5} \ac{ssm} with the configuration used for the study reported in Fig.~\ref{fig:nonlinear_noise_sensitivity}, with the exception that here we set \( \sigma_{\evec} = 0.1, ~\sigma_{\vvec}=\sqrt{2} \), and evaluate performance vs different graph parameters in Fig.~\ref{fig:nonlinear_summary_all}. 

Figures~\ref{fig_mse_vs_sparsity_nonlin2} and \ref{fig_eier_vs_sparsity_nonlinear2} plot the \ac{mse} and \ac{eier} versus the sparsity level of the true graph used at initialization. 
In Fig.~\ref{fig_mse_vs_sparsity_nonlin2}, all methods show higher \ac{mse} as the graph becomes denser.  
This rise is mainly due to an error-aggregation effect, as having more active edges means more error terms are added, causing the total (or average) \ac{mse} to increase.
Change-det scheme shows greater sensitivity in this setting, as it is designed for sparse changes. With many active edges and large process noise, the edge-weight updates become larger than what the model is designed to handle, leading to increased error.

In Fig.~\ref{fig_eier_vs_sparsity_nonlinear2}, the \ac{eier} for the \ac{gsp}-\ac{ekf}, \ac{ekf}, and Change-det methods decreases as the graph becomes denser.  
 As the number of true edges increases,  accurate edge estimation becomes more likely, leading to a lower \ac{eier}.
Notably, the \ac{gsp}-\ac{ekf} achieves lower \ac{eier} due to its awareness of sparsity, which reduces the number of falsely estimated connections. 
Oracle-\ac{gsp}-\ac{ekf}  maintains near-zero \ac{eier}, with a slight increase at denser graphs. This increase is linked to the rise in \ac{mse} observed in Fig.~\ref{fig_mse_vs_sparsity_nonlin2}, which raises the likelihood of negative weight estimates that are subsequently set to zero after the update step.



We continue by examining the sensitivity of each method to the percentage of connection changes between consecutive support updates, as shown in Figs. \ref{fig_mse_vs_deltaN_nonlin2} and \ref{fig_eier_vs_deltaN_nonlinear2}.
The number of time samples between consecutive support updates is fixed at $3N$, and the tuning parameter $\mu$ for the \ac{gsp}-\ac{ekf} is adjusted according to the percentage of changed connections: 
$\mu = 0.25$ for less than 7.5\%, $\mu = 0.15$ for 7.5\%–15\%, and $\mu = 0.1$ for more than 15\%.

In Fig.~\ref{fig_mse_vs_deltaN_nonlin2}, we observe that for the \ac{gsp}-\ac{ekf}, \ac{ekf}, and Change-det methods, the \ac{mse} increases with the number of changed connections. This is because larger changes introduce greater initial adaptation errors and require estimating more parameters within the same number of time samples. The Change-det method exhibits the highest \ac{mse} in this setting, as it is tailored for sparse edge changes and assumes a less nonlinear measurement model. 
Moreover, as the number of changed connections increases, the true state becomes denser. This trend helps explain, in light of the influence of the sparsity
level of the true graph analysis (Figs.~\ref{fig_mse_vs_sparsity_nonlin2} and \ref{fig_eier_vs_sparsity_nonlinear2}), both the increase in \ac{mse} and the decrease in \ac{eier} observed in Fig.~\ref{fig_eier_vs_deltaN_nonlinear2} for the Change-det and \ac{ekf} methods. 
The \ac{eier} of the \ac{gsp}-\ac{ekf} increases due to the reduction in the threshold value $\mu$, which is adjusted to account for changing sparsity. Nevertheless, it remains the second-best performer after the Oracle-\ac{gsp}-\ac{ekf}.

We continue by examining the sensitivity of each method to the number of time samples between consecutive support updates, as shown in Figs. \ref{fig_mse_vs_k_nonlin2} and \ref{fig_eier_vs_k_nonlinear2}.
The tuning parameter 
$\mu$ for the \ac{gsp}-\ac{ekf} is adjusted based on this interval: 
$\mu=0.1$ for fewer than 
$2^3$ samples and $\mu=
0.25$ otherwise.
As expected, both \ac{mse} (Fig. \ref{fig_mse_vs_k_nonlin2}) and \ac{eier} (Fig. \ref{fig_eier_vs_k_nonlinear2}) decrease as the number of time samples between consecutive support updates increases. This is due to the greater availability of measurements under a fixed support, which enables more accurate estimation. 
 For the \ac{gsp}-\ac{ekf}, \ac{ekf}, and Change-det methods
the significant drop in \ac{mse} and \ac{eier} occurs when the amount of time units between consecutive changes smaller than $4.5$, which is consistent with Theorem~\ref{observbility}, that states that observability is not guaranteed below the threshold of $4.5$. 
 For  Oracle-\ac{gsp}-\ac{ekf}, both metrics remain approximately constant, as it only needs to estimate the weights of the active edges, 
and even a small number of measurements 
is sufficient for accurate estimation.  


\begin{figure*}[t]
\centering
\begin{subfigure}[t]{0.3\textwidth}
    \includegraphics[width=\linewidth]{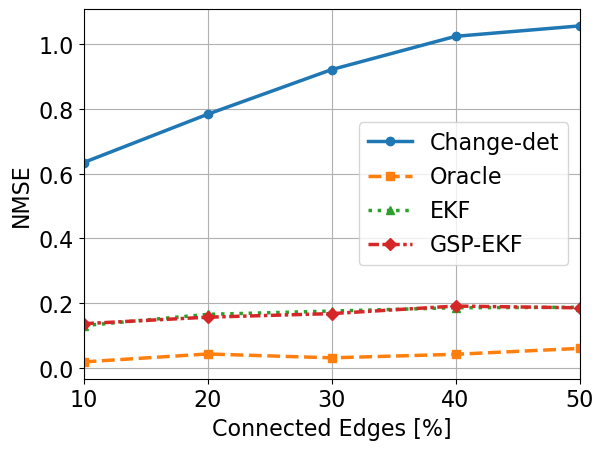}
    \caption{}
    \label{fig_mse_vs_sparsity_nonlin2}
\end{subfigure}
\begin{subfigure}[t]{0.3\textwidth}
    \includegraphics[width=\linewidth]{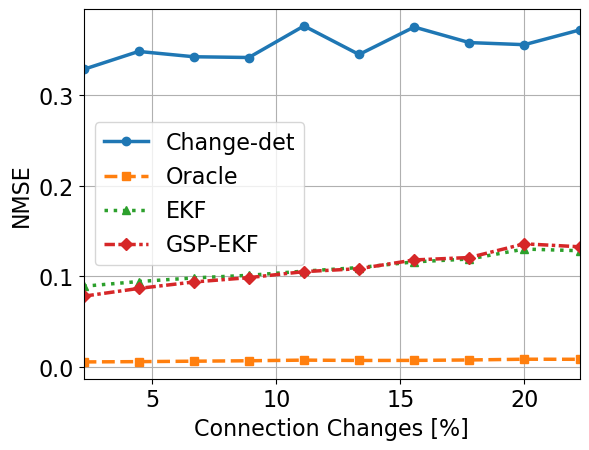}
    \caption{}
    \label{fig_mse_vs_deltaN_nonlin2}
\end{subfigure}
\begin{subfigure}[t]{0.3\textwidth}
    \includegraphics[width=\linewidth]{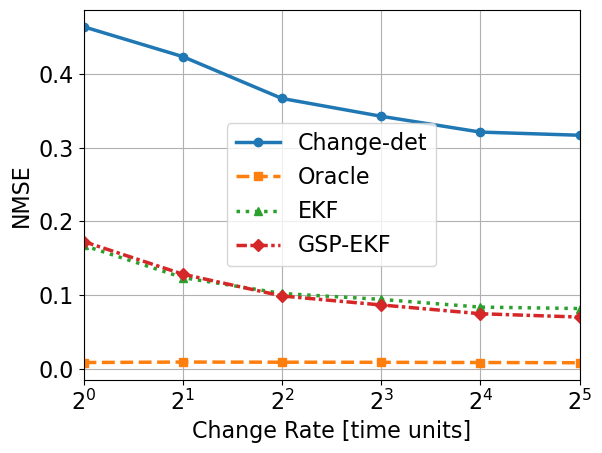}
    \caption{}
    \label{fig_mse_vs_k_nonlin2}
\end{subfigure}
\vspace{0.3cm}

\begin{subfigure}[t]{0.3\textwidth}
    \includegraphics[width=\linewidth]{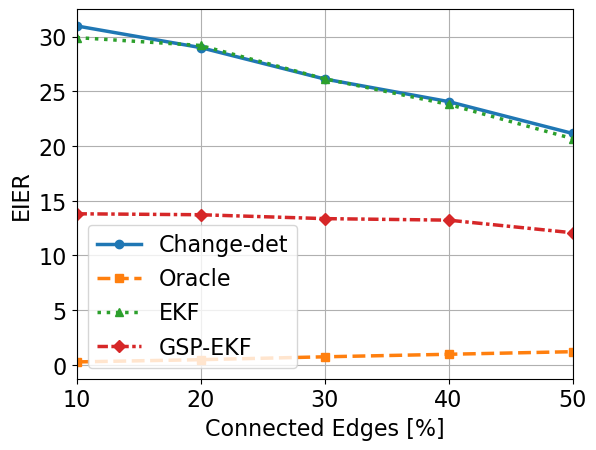}
    \caption{}
    \label{fig_eier_vs_sparsity_nonlinear2}
\end{subfigure}
\begin{subfigure}[t]{0.3\textwidth}
    \includegraphics[width=\linewidth]{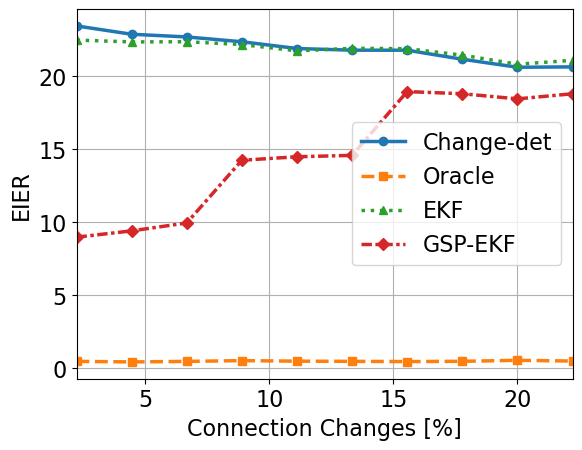}
    \caption{}
    \label{fig_eier_vs_deltaN_nonlinear2}
\end{subfigure}
\begin{subfigure}[t]{0.3\textwidth}
    \includegraphics[width=\linewidth]{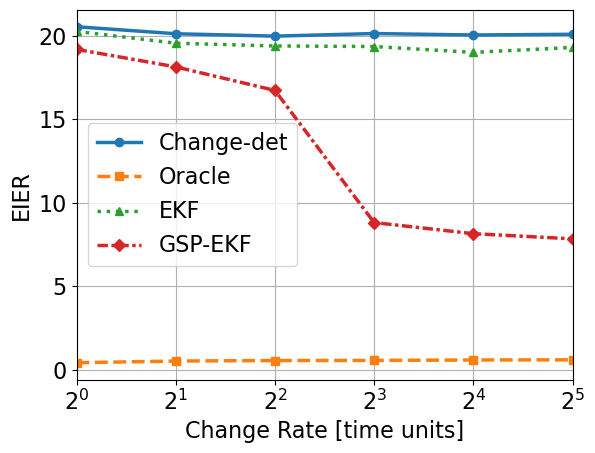}
    \caption{}
    \label{fig_eier_vs_k_nonlinear2}
\end{subfigure}
\caption{\ac{mse} (top) and \ac{eier} (bottom) versus: (a), (d) the sparsity level of the true graph used at initialization; (b), (e) the number of edge changes between consecutive support updates; and (c), (f) the number of time samples between consecutive edge set updates, in the {\em NL5} \ac{ssm}.
}
\label{fig:nonlinear_summary_all}
\end{figure*}

\subsubsection{Runtime Comparison}\label{ssec:sim_time}
We proceed to investigate how the computational cost and estimation performance of the \ac{gsp}-\ac{ekf} are affected by the complexity of the filter order, $P$, in the nonlinear measurement model in the {\em NLP} \ac{ssm}. 
We set the simulation parameters similar to the configuration used for the study reported in Fig.~\ref{fig:nonlinear_summary_all}.
The \ac{gsp}-\ac{ekf} parameter  is set to 
$\mu=0.25$ for $P<7$ and $\mu=0.15$ for $P>7$.


Fig.~\ref{fig_times_vs_polyOrder} shows that the average computation time per incoming measurement increases for all methods as \( P \) increases, since the complexity of evaluating both \( \hvec_{\timeIndx}(\xvec_{\timeIndx}) \) and its Jacobian grows polynomially with \( P \).  
The \ac{gsp}-\ac{ekf} and  Oracle-\ac{gsp}-\ac{ekf} achieve the lowest computation times, as both leverage the efficient Jacobian computation described in Algorithm~\ref{alg:DP_Jacobian}, which scales as \( O(PN^3) \). 
Next is the \ac{ekf}, which does not employ this optimization and instead scales as \( O(P^3N^4) \). The difference in computational complexity \ac{wrt} \( P \) explains why, for \( P = 1 \), the runtime differences among the \ac{gsp}-\ac{ekf}, \ac{ekf}, and Oracle-\ac{gsp}-\ac{ekf} are relatively small, and why the run time gap between the \ac{ekf} and the other two methods becomes more pronounced as \( P \) increases.  

Finally, Change-det exhibits the highest runtime, as it involves solving an optimization problem using the CVX package in Python, which becomes significantly slower to converge in the nonlinear measurement setting.

Figures~\ref{fig_mse_vs_polyOrder} and~\ref{fig_eier_vs_polyOrder} show that both \ac{mse} and \ac{eier} increase with the filter order \( P \), as the estimation becomes more nonlinear and therefore more challenging. The relative performance trends among the methods remain consistent with earlier simulations. 
Change-det exhibits a steep increase in \ac{mse} due to model mismatch, which worsens as data nonlinearity increases.




\begin{figure*}[hbt]
\captionsetup[subfigure]{labelformat=empty}
     \centering
     \begin{subfigure}[b]{0.31\textwidth}
         \centering
         \includegraphics[width=\textwidth]{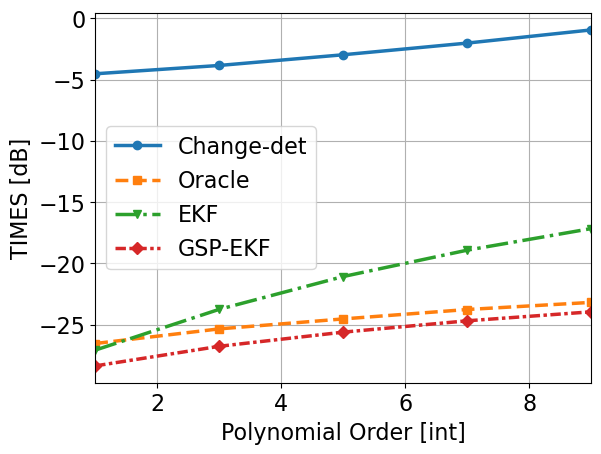}
         \caption{(a) Computation time versus the filter order $P$.}
         \label{fig_times_vs_polyOrder}
     \end{subfigure}
     \hfill
     \begin{subfigure}[b]{0.31\textwidth}
         \centering
         \includegraphics[width=\textwidth]{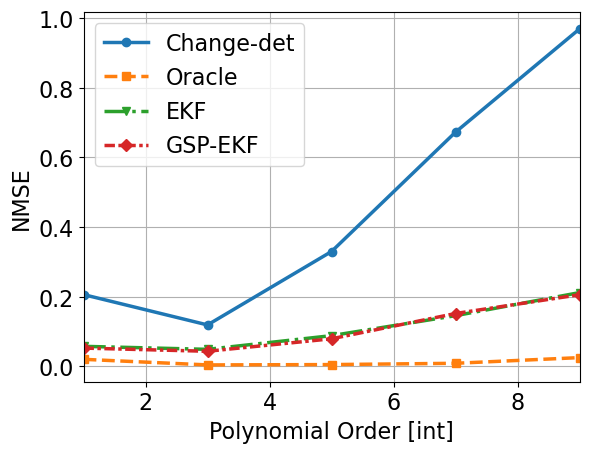}
         \caption{(b) N\ac{mse} versus the  filter order $P$.}
         \label{fig_mse_vs_polyOrder}
     \end{subfigure}
     \hfill
     \begin{subfigure}[b]{0.31\textwidth}
         \centering
         \includegraphics[width=\textwidth]{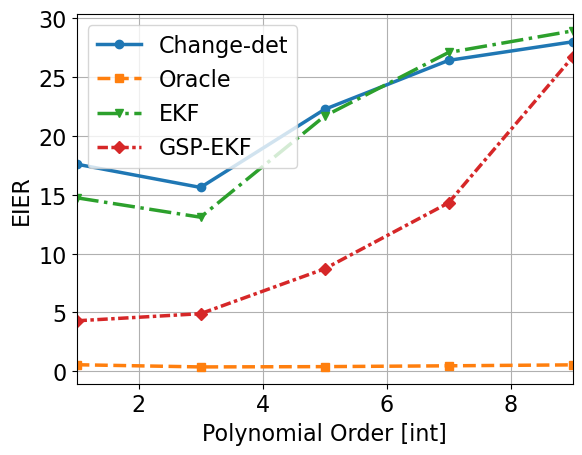}
         \caption{(c) \ac{eier}  versus the filter order $P$.}
         \label{fig_eier_vs_polyOrder}
     \end{subfigure}
        \caption{Performance versus the filter order $P$, {\em NLP} \ac{ssm}.}
        \label{fig:poly_order}
\end{figure*}
\section{Conclusions}\label{sec:conclusions}

In this paper we propose a novel framework for tracking dynamic graph topologies by formulating a sparse nonlinear \ac{ssm} under the \ac{gsp} paradigm. The evolving topology was modeled through time-varying edge weights and unknown graph support, where the Laplacian matrix was parameterized via the incidence matrix. The resulting measurement model, derived from a graph filter, is inherently nonlinear and topology-dependent, necessitating the use of an \ac{ekf} for state estimation.
To handle the sparsity of networks, we incorporated $\ell_1$-regularization into the EKF update step, yielding a sparsity-aware GSP-EKF algorithm. We derived a closed-form expression for the Jacobian of the graph filter and introduced an efficient dynamic-programming-based computation method that reduces the complexity of its computation. The properties of the proposed \ac{gsp}-\ac{ekf} method were investigated, and some special cases are presented. 
Through numerical experiments on both linear and nonlinear measurement models, the proposed \ac{gsp}-\ac{ekf} method demonstrated superior performance in terms of edge identification error and MSE. 
It consistently outperformed standard \ac{ekf} and Change-det methods, and closely approached the performance of  Oracle-\ac{gsp}-\ac{ekf} estimator with known support.
Overall, this work lays the foundation for robust and efficient topology tracking in dynamic networks within the graph signal processing framework.

\appendices
\renewcommand{\thesectiondis}[2]{\Alph{section}:}
\section{Proof of Theorem \ref{observbility}}\label{app:proof_observability}
In this appendix, we analyze the noiseless 
 linear time-varying system in \eqref{eq_model2}
over the interval \( k = \timeIndx, \timeIndx+1, \dots, \timeIndx+T-1 \), and identify conditions under which the initial state \( \xvec_{\timeIndx} \in \mathbb{R}^{\maxEdges} \) can be uniquely recovered from the observations $\{\yvec_{k}\}_{k=\timeIndx}^{\timeIndx+T-1}$ and known inputs $\{\cvec_{k}\}_{k=\timeIndx}^{\timeIndx+T-1}$. 
 As in standard observability analysis (e.g., \cite[Ch. 24.3]{dahleh2004lectures},  we assume \( \cvec_k = \zerovec \) without loss of generality, since any nonzero input can be absorbed into the observations by defining a modified vector \( \yvec_k - \cvec_k \). 
 
By  recursively substituting the state evolution equations into the measurement model, 
 each observation in \eqref{eq_model2}   can be expressed as a function of the initial state. Stacking these equations yields the following linear system:
\begin{eqnarray}\label{observability_proof}
    \left[
    \yvec_{\timeIndx}^T
    \ldots,
    \yvec_{\timeIndx+T-1}^T \right]^T
=\mathbf{O}_{\timeIndx}^{(T)}\xvec_{\timeIndx},
\end{eqnarray}
where the last equality is obtained 
 by substituting the definition of the $T$-step observability matrix $\mathbf{O}_{\timeIndx}^{(T)}$ from \eqref{matrix_O}.

The question of $T$-step observability at time $\timeIndx$ reduces to determining whether 
$\xvec_\timeIndx$ can be uniquely identified from the measurements 
$\{\yvec_{k}\}_{k=\timeIndx}^{\timeIndx+T-1}$ \cite[Ch. 24.3]{dahleh2004lectures}. From \eqref{observability_proof}, this holds if and only if the observability matrix  $\mathbf{O}_{\timeIndx}^{(T)}$ has full column rank, that is, $\text{rank}(\mathbf{O}_{\timeIndx}^{(T)}) = \maxEdges$.
\section{Proof of Lemma \ref{lemma_jacobian}}\label{app;jacobian_calc}
To prove Lemma~\ref{lemma_jacobian}, we compute the derivative of the graph filter output in \eqref{eqn:GraphFiltRep} \ac{wrt} the $m$th element of the edge-weight vector, $[\xvec_{\timeIndx}]_m$. First, by applying the partial derivative and using the linearity of differentiation, we obtain the $m$th column: 
\beqna\label{jacobain_derivation_step1}  
[\nabla_{\xvec_{\timeIndx}}\hvec_{\timeIndx}(\xvec_{\timeIndx})]_{:,m}=\frac{\partial  \hvec_{\timeIndx}(\xvec_{\timeIndx})}{\partial [\xvec_{\timeIndx}]_m} = 
\sum_{p=1}^P a_p \frac{\partial \Lmat_{\timeIndx}^p\qvec_{\timeIndx}}{\partial [\xvec_{\timeIndx}]_m}.
\eeqna   
Applying the product rule for matrix differentiation  
(see Equation (37) in \cite{matrix_cookbook}) and noting that $\qvec_{\timeIndx}$ is independent of $\xvec_{\timeIndx}$, we can write
\beqna\label{jacobain_derivation_step2}  
\frac{\partial \Lmat_{\timeIndx}^p\qvec_{\timeIndx}}{\partial [\xvec_{\timeIndx}]_m}
=\sum_{k=0}^{p-1} \Lmat_{\timeIndx}^k 
\frac{\partial \Lmat_{\timeIndx}}{\partial [\xvec_{\timeIndx}]_m} 
\Lmat_{\timeIndx}^{p-k-1}\qvec_{\timeIndx}.
\eeqna  
Substituting the Laplacian expression from \eqref{L2B_time_index} into \eqref{jacobain_derivation_step2}, and using the facts that $\Bmat$ is independent of $\xvec_{\timeIndx}$ and the identity
$\frac{\partial \xvec_{\timeIndx}}{\partial [\xvec_{\timeIndx}]_m}=\uvec_m$, where $\uvec_m$ is the $m$th column of the $\maxEdges\times\maxEdges$ identity matrix, whose entries are given by $[\uvec_m]_i = 1$ if $i=m$ and $[\uvec_m]_i = 0$ if $i\neq m$.
we obtain  
\beqna\label{jacobain_derivation_step3}  
\frac{\partial \Lmat_{\timeIndx}^p\qvec_{\timeIndx}}{\partial [\xvec_{\timeIndx}]_m}
=\sum_{k=0}^{p-1} \Lmat_{\timeIndx}^k \Bmat\diag(\uvec_m) \Bmat^T 
\Lmat_{\timeIndx}^{p-k-1}\qvec_{\timeIndx}.
\eeqna  
Finally, substituting \eqref{jacobain_derivation_step3} into \eqref{jacobain_derivation_step1}, 
and writing $\Bmat \diag(\uvec_m) \Bmat^T$ as $\Bmat_{:,m} \Bmat_{:,m}^T$ results in 
 \eqref{jacobain_poly}, which completes the proof.  
\section{Proof of Lemma~\ref{Lemma_jacobian_reindex}}\label{app;jacobian_reindex}
We begin with the expression for the $m$th column of the Jacobian matrix in \eqref{jacobain_poly}, and perform a change of variables $ p' = p - 1 $ and $ j = p- 1- k =p'-k $. Since $k$ ranges from $0$ to $p'=p-1$,
 it follows that $j$ also ranges from $0$ to $p'$.
Substituting this change of variables in \eqref{jacobain_poly} yields
\begin{equation}  
\label{jacobain_poly2}  
[\nabla_{\xvec_{\timeIndx}}\hvec_{\timeIndx}(\xvec_{\timeIndx})]_{:,m}=
\sum_{p'=0}^{P-1} a_{p'+1} \sum_{j=0}^{p'} \Lmat_{\timeIndx}^{p'-j} 
\Bmat_{:,m} \Bmat_{:,m}^T
\Lmat_{\timeIndx}^{j}\qvec_{\timeIndx}.
\end{equation}
Since $ j =p'-k$ and $ 0 \leq j \leq p' \leq P-1 $, we can change the order of summation to obtain
\beqna
\label{jacobain_poly_reordered}
[\nabla_{\xvec_{\timeIndx}}\hvec_{\timeIndx}(\xvec_{\timeIndx})]_{:,m}\hspace{6cm}\nonumber\\=
\sum_{m=0}^{P-1} \left(\sum_{p'=j}^{P-1} 
a_{p'+1} 
\Lmat_{\timeIndx}^{p'-j}\right) 
\Bmat_{:,m} \Bmat_{:,m}^T
\Lmat_{\timeIndx}^{j} 
\qvec_{\timeIndx}.
\eeqna
To further simplify the expression and make it suitable for dynamic programming, we define a new variable $ r = p' - j $, such that $ p' = j + r $ and $r$ 
ranges from $0$ to $P - 1 - j$. 
Substituting this into the inner summation in \eqref{jacobain_poly_reordered}, yields the results in \eqref{jacobain_poly_reindexed}.
\section{Proof of Claim \ref{claim;starting_point}}\label{app;proof_initialization}
By substituting $\hat{\xvec}_{\timeIndx}$ from \eqref{eqn:EKFUpdate}-\eqref{eqn:kalman_gain_computaion}  into \eqref{gradient_phi1}, we obtain
 \beqna\label{gradient_phi_substitute}
\nabla_\xvec\phi_l(\hat{\xvec}_{\timeIndx})= 
\Bigg(\left(\hat{\Hmat}_{\timeIndx}^T\Rmat^{-1}\hat{\Hmat}_{\timeIndx}+\hat{\bf{\Sigma}}_{\timeIndx|\timeIndx-1}^{-1}\right)\hat{\bf{\Sigma}}_{\timeIndx|\timeIndx-1}\hat\Hmat_{\timeIndx}^T \hat{\Smat}_{\timeIndx|\timeIndx-1}^{-1}\nonumber\\-\hat{\Hmat}_{\timeIndx}^T\Rmat^{-1}\Bigg)(\yvec_{\timeIndx}-\hvec_{\timeIndx}(\hat{\xvec}_{\timeIndx|\timeIndx-1})).
\eeqna
By using the variant of the Woodbury identity for positive definite matrices $\Pmat,\Rmat$ (see Equation (80) in \cite{Welling2010})
\begin{equation}
\label{eq:158}
\bigl(\Pmat^{-1} + \Bmat^{T} \Rmat^{-1} \Bmat \bigr)^{-1} \Bmat^{T} \Rmat^{-1}
\;=\;
\Pmat \, \Bmat^{T}
\bigl(\Bmat\, \Pmat \, \Bmat^{T} + \Rmat \bigr)^{-1},
\end{equation}
 along with the definition of ${\hat\Smat}_{\timeIndx|\timeIndx-1}$ from \eqref{eqn:obs_covariance_computaion} one obtains 
 \beqna \label{identitiy_claim2}
\left(\hat{\Hmat}_{\timeIndx}^T\Rmat^{-1}\hat{\Hmat}_{\timeIndx}+\hat{\bf{\Sigma}}_{\timeIndx|\timeIndx-1}^{-1}\right)\hat{\bf{\Sigma}}_{\timeIndx|\timeIndx-1}\cdot\hat\Hmat_{\timeIndx}^T\cdot \hat{\Smat}_{\timeIndx|\timeIndx-1}^{-1}=\hat{\Hmat}_{\timeIndx}^T\Rmat^{-1}.
 \eeqna 
Substituting \eqref{identitiy_claim2} into \eqref{gradient_phi_substitute} yields \( \nabla_{\xvec} \phi_{\timeIndx}(\hat{\xvec}_{\timeIndx}) = \zerovec \), and accordingly, \eqref{proxy_solution} can be expressed as \eqref{proxy_solution_1_iter}. Comparing this result (obtained for the linear \ac{ssm} in Example~\ref{sssec:LinearExample}) with the GT-KF2 algorithm from~\cite{dabush2024Routtenberg_Kalman} reveals that the two methods are equivalent in this special case.
\bibliographystyle{IEEEtran}
\bibliography{IEEEabrv,refs}

\end{document}